\documentclass[11pt]{article}
\usepackage{amflow}

\usepackage{setspace}
\usepackage{authblk}

\title{Faster Weak Expander Decompositions and Approximate Max Flow}
\author[1]{Henry Fleischmann}
\author[1]{George Z. Li}
\author[1]{Jason Li}
\affil[1]{Carnegie Mellon University $\{\texttt{hfleisch}, \texttt{gzli}, \texttt{jmli}\}$\texttt{@cs.cmu.edu}}
\date{}

\begin{document}

\maketitle

\begin{abstract}
We give faster algorithms for weak expander decompositions and approximate max flow on undirected graphs. First, we show that it is possible to ``warm start’’ the cut-matching game when computing weak expander decompositions, avoiding the cost of the recursion depth. Our algorithm is also flexible enough to support weaker flow subroutines than previous algorithms. 

Our second contribution is to streamline the recent non-recursive approximate max flow algorithm of Li, Rao, and Wang (SODA, 2025) and adapt their framework to use our new weak expander decomposition primitive. Consequently, we give an approximate max flow algorithm within a few logarithmic factors of the limit of expander decomposition-based approaches. 
\end{abstract}

\newpage 
\tableofcontents
\newpage

\section{Introduction}

In the maximum flow problem, we are given a set of vertex demands, where each vertex is required to send or receive a certain amount of flow, and the goal is to route these demands while minimizing the maximum congestion along any edge. It is one of the oldest problems in theoretical computer science~\cite{dantzig1951application}, with surprising connections to other famous problems including minimum cut, bipartite matching, and Gomory-Hu trees~\cite{gomory1961multi}. And yet, despite its status as a notoriously difficult problem, modern algorithmic techniques have produced exciting breakthroughs in both the exact directed setting~\cite{daitch2008faster,madry2016computing,chen2025maximum} and the approximate undirected setting~\cite{lee2013new,sherman2013nearly,sherman2017area}. These techniques include the interior point method from continuous optimization~\cite{daitch2008faster,madry2016computing,chen2025maximum}, electrical flows and Laplacian solvers~\cite{spielman2014nearly,lee2013new,madry2016computing}, expander decompositions and congestion approximators~\cite{sherman2013nearly,racke2014computing}, and dynamic data structures~\cite{chen2025maximum}.

On the other hand, despite the rapid advancement of modern flow algorithms, progress towards \emph{understanding} max flow, especially its underlying structural properties, has arguably lagged behind:
\begin{enumerate}
\item The state of the art $(1-\epsilon)$-approximate max flow algorithm on undirected graphs runs in $O(m\log^{41}n\log^2\log n)$ time even for constant $\epsilon>0$~\cite{peng2016approximate}. The algorithm is fairly complex, recursively alternating between multiple different problems, and it took a decade before the first non-recursive algorithm was developed for this problem~\cite{DBLP:conf/soda/0006R025}.
\item While the exact max flow algorithms based on interior point methods are impressive, they do not shed light on the \emph{combinatorial} structure of max flow. In response, a recent trend of studying combinatorial max flow has emerged~\cite{chuzhoy2024faster,chuzhoy2024maximum,bernstein2024maximum,bernstein2025combinatorial}, obtaining augmenting path-style algorithms that are more faithful to traditional approaches.
\end{enumerate}

This paper is dedicated to improving our understanding of approximate max flow in the undirected setting. Our starting point is the recent non-recursive algorithm for approximate max flow~\cite{DBLP:conf/soda/0006R025}, which computes a hierarchy of so-called \emph{weak expander decompositions}, using previously computed levels of the hierarchy to build the next level. From this hierarchy, a \emph{congestion approximator} is extracted and used in Sherman's framework~\cite{sherman2017area} to obtain the desired approximate max flow. However,~\cite{DBLP:conf/soda/0006R025} do not state an explicit running time, since the weak expander hierarchy construction requires calls to \emph{fair cut/flow}~\cite{DBLP:conf/soda/0006NPS23}, which introduce a large running time overhead.

The contribution of this paper is twofold:
\begin{enumerate}
\item First, we develop a faster weak expander decomposition algorithm by ``warm starting'' the cut-matching game~\cite{DBLP:journals/jacm/KhandekarRV09} whenever a sparse cut is found. This algorithm can be implemented using $O(\log^2n)$ calls to max flow, compared to $O(\log^3n)$ for the standard weak expander decomposition implementation~\cite{saranurak2019expander}, and may be of independent interest.
\item Next, we streamline the framework of~\cite{DBLP:conf/soda/0006R025} to obtain a non-recursive approximate max flow algorithm with an improved running time of $O(m\log^9n\log\log n)$. In particular, we implement our weak expander decomposition using approximate max flow, compared to prior algorithms which require fair cut/flow or similarly strong guarantees~\cite{saranurak2019expander,DBLP:conf/soda/0006R025}. Similar to~\cite{DBLP:conf/soda/0006R025}, these max flow calls are specialized enough to be solvable using the existing levels of the hierarchy. However, a weaker flow oracle introduces a number of technical difficulties which we discuss in the technical overview.
\end{enumerate}

\begin{theorem}[Informal version of \Cref{cor:main-flow-result}]
Given an undirected graph with integral and polynomially-bounded edge capacities, there is an $O(m\log^{9}n)$ time algorithm to construct a congestion-approximator with quality $O(\log^5n)$. Together with Sherman's framework~\cite{sherman2017area}, we obtain an $(1-\epsilon)$-approximate max flow algorithm in time $O(m\log^9n\log\log n+\epsilon^{-1} m\log^6n)$.
\end{theorem}

While the logarithmic exponent of $9$ is too large to be practical, we remark that expander decomposition-based algorithms have historically led to similarly large constants. For example, the state of the art (strong) $\phi$-expander decomposition~\cite{saranurak2019expander} deletes $O(\phi\log^3n)$ fraction of edges and runs in time $O(m\log^5n/\phi)$ on capacitated graphs; to delete a constant fraction of edges, we require $\phi\approx1/\log^3n$ which results in $O(m\log^8n)$ time. Even our faster implementation of weak expander decomposition runs in $O(m\log^4n)$ time in the most ideal setting. Therefore, our max flow algorithm is within a few logarithmic factors of the limit to any expander decomposition-based approaches, and substantial future improvements will require either breakthroughs in computing expander decompositions, or bypassing expander decompositions altogether.

\section{Technical Overview}

\paragraph{Faster weak expander decompositions.} All known algorithms for computing expander decompositions in near-linear time rely on the cut-matching game. Our first main technical contribution is a faster algorithm for computing weak expander decompositions. We do this by observing that it is possible to ``warm start'' our recursive instances of the cut-matching game. Importantly for our application to approximate max flow, our algorithm is robust enough to support general vertex weights and to implement the matching steps using approximate max flow oracles. 

We now describe the techniques in more detail. In the standard cut-matching game (on an unweighted graph), we have $T=\Theta(\log^2n)$ rounds in total. In each round, the cut player finds two disjoint sets $L_A$ and $R_A$. The matching player then tries to route a flow from $L_A$ to $R_A$, implicitly defining a matching between the sets. If at any iteration the matching player fails to route the flow, the cut certifying infeasibility of the flow is a sparse cut, showing that the graph is not an expander. Otherwise, if all matching step flows are feasible, the cut player is defined so that the union of the matchings found in the $T$ iterations is itself an expander. Combined with the fact that the matchings embed into the original graph with low congestion, this proves that $G$ must be an expander.

In the non-stop version of the cut-matching game~\cite{racke2014computing,saranurak2019expander}, when the matching player fails to route the flow from $L_A$ to $R_A$, thus finding a sparse cut $S$, the algorithm does not (necessarily) immediately terminate. Instead, the cut-matching game continues on $V\setminus S$. More generally, let $A$ be the current set on which the cut-matching game is being played; when the matching player finds a cut $S$, the algorithm continues on $A\setminus S$. It can be shown that after $T$ iterations, the remaining set $A$ is a near-expander in $G$, meaning that its degree vertex weighting mixes in $G$ (but possibly not $G[A]$) with low congestion.

To convert the non-stop cut-matching game into a weak expander decomposition algorithm, \cite{directed-expander-decomposition} adds an \textit{early termination condition}: if $\vol(A)\le 99\vol(V)/100$ at any point, the non-stop cut-matching game terminates and recurses on $A$ and $V\setminus A$. Otherwise, if the cut-matching game terminates in certifying that $A$ is a near-expander, we recurse onto $V \setminus A$ if it is non-empty. If $A$ is certified as a near-expander in some iteration, we then know that $\vol(V\setminus A)\le \vol(V)/100$ so the recursive call decreases by a constant factor in size. If we reach the early termination condition, then $\vol(A)\le 99\vol(V)/100$ and we have the additional guarantee that $V \setminus A$ is partitioned by sparse cuts into subsets of at most $2/3$ of the volume each. Hence, the recursive calls also decrease in size by a constant factor in this case. As a result, the recursive depth is at most $O(\log{n})$, so we can compute a weak expander decomposition in $O(\log^3{n})$ iterations of the cut-matching game.

We give a new weak expander decomposition algorithm which only uses $O(\log^2{n})$ iterations of the cut-matching game. Let $A$ again be the current set of vertices. If we find a matching successfully, then we continue the cut-matching game on $A$. Otherwise, we find a sparse cut $S\subseteq A$. In the previous algorithm, we would only continue on $A\setminus S$, delaying continuing on $S$ until reaching the early termination condition or certifying that some subset of $A \setminus S$ is a near-expander in $G$. Our main observation is that we can continue the cut-matching game on both $A\setminus S$ and $S$ simultaneously without a loss in the runtime. This amounts to ``warm-starting'' on $S$. Slightly more formally, we maintain a partition of $V$ into sets $\mathcal{A}_t=\{A_1,\ldots,A_{k_t}\}$ at each iteration $t$. At the beginning of the algorithm, we set $\mathcal{A}_0=\{V\}$ and at each iteration $t$, we run the cut-matching game on each $A_i \in \mathcal{A}_t$ simultaneously. When we find some cut $S_i\subseteq A_i$, we add $S_i$ (if nonempty) and $A_i\setminus S_i$ to $\mathcal{A}_{t+1}$. After $T$ rounds, we will certify that each component in $\mathcal{A}_T$ is a near-expander. We remark that warm-stating crucially uses the fact that we are ultimately constructing a weak expander decomposition, not a strong one. Indeed, the matching embeddings from steps prior to restricting to a subgraph (from finding a cut) are not guaranteed to embed into our current subgraph. This is fine for certifying near-expansion but too weak a guarantee for strong expansion.

Importantly for our application to approximate max flow, this algorithm still reveals sufficient structure when implementing the matching steps with an approximate max flow oracle. To this end, we give our algorithm in two steps. The first step is to compute a weak expander decomposition where there can be a small ``deleted,'' non-expanding portion of the input vertex weighting (\cref{sec: weak expander wdd}). The non-deleted portion is certified to mix simultaneously (i.e., each component expands with respect to the non-deleted portion of the vertex weighting) and there are guaranteed to be few intercluster edges, as usual. Then, in the second step (\cref{sec: grafting}), we attempt to \textit{graft} the demand deleted in each cluster back into the cluster, as in~\cite{directed-expander-decomposition}. After this step, every expanding cluster will not have any deleted demand, and nearly all demand will belong to an expanding cluster. We state an informal version of our result in~\Cref{thm:informal-weak}.
\begin{theorem}[Informal version of~\Cref{thm: weak-decomp-del2}]\label{thm:informal-weak}
Suppose we have $G=(V, E, \bc)$ with integer edge capacities at most $\poly(n)$. In addition, suppose we have  vertex weighting $\bd \in \Z_{\geq 0}^V$, expansion parameter $\phi > 0$, and a suitable approximate max flow oracle running in time $F(n,m, \eps)$. 
Then, there is an algorithm computing a partition $\calA = \calA^\circ \sqcup \calA^\times$ of $V$ with the following properties:
\begin{enumerate}
 \item The algorithm runs in time $O(F(n,m, \eps) \log^2 n + m\log^4 n)$.
      \item $\bd(\bigcup_{A \in \calA^\times} A) = O((\eps \log^2 n + \phi \log n)\bd(V)).$ 
        \item The total capacity of edges cut by $\calA$ is at most $O(\phi \bd(V) \log n )$.
    \item Each $A \in \calA^\circ$ is a $(\phi/\log^2 n, \bd)$-near-expander in $G$.
\end{enumerate}
\end{theorem}
Importantly for our applications, we actually obtain a stronger simultaneous mixing expansion property instead of (4), but we omit that here for simplicity (see~\Cref{thm: weak-decomp-del2}). Also note that, unlike standard weak expander decompositions, our result does not exactly decompose all vertices into near-expanders (i.e., usual decompositions would get $\calA^\times = \emptyset$ or the guarantee of (2) to be $0$). However, this relaxation is critical for obtaining such a result and still suffices for some important applications of weak expander decompositions. Indeed, the relaxation of (2) suffices for our application to constructing congestion-approximators and approximate max flow, as we discuss next.

\paragraph{Faster congestion-approximators.}

Recall that a laminar family $\calC$ of subsets of $V$ forms an $\alpha$-congestion-approximator if for every vertex demand, the minimum ratio over cuts $C\in\cal C$ between the capacity of the cut and the demand crossing the cut is an $\alpha$-approximation to the optimal congestion of any flow routing the demand. Our goal is to construct $\alpha$-congestion-approximators faster and with smaller $\alpha$.
To discuss our improvement over previous work, we restate the informal Theorem 2.1 from \cite{DBLP:conf/soda/0006R025}, which gave a novel approach for constructing congestion-approximators.
\begin{theorem}[Theorem 2.1 of~\cite{DBLP:conf/soda/0006R025}]
    Consider a capacitated graph $G=(V,E, \bc)$, and let $\alpha\ge 1$ and $\beta\ge1$ be parameters. Suppose there exist partitions $\mathcal{P}_1,\mathcal{P}_2,\ldots,\mathcal{P}_L$ of $V$ such that
	\begin{enumerate}
		\item $\mathcal{P}_1$ is the partition $\{\{v\}:v\in V\}$ of singleton clusters, and $\mathcal{P}_L$ is the partition $\{V\}$ with a single cluster.
		\item For each $i\in[L-1]$, for each $C\in\mathcal{P}_{i+1}$, the intercluster edges of $\mathcal{P}_{i}$ internal to $C$ along with the boundary edges of $C$ mix in the graph $G$. Moreover, the mixings over all the clusters $C\in\mathcal{P}_{i+1}$ have congestion $\alpha$ simultaneously.
		\item For each $i\in[L-1]$, there is a flow in $G$ with congestion $\beta$ such that each intercluster edge of $\mathcal{P}_{i+1}$ sends its capacity in flow, and each intercluster edge of $\mathcal{P}_i$ receives half its capacity in flow.
	\end{enumerate}
	For each $i\in[L]$, let partition $\mathcal{R}_{\ge i}$ be the common refinement of partitions $\mathcal{P}_i,\mathcal{P}_{i+1},\ldots,\mathcal{P}_L$, i.e.,
	\begin{align*}
		\mathcal{R}_{\ge i}=\{C_i\cap \cdots\cap C_L:C_i\in\mathcal{P}_i,\ldots,C_L\in\mathcal{P}_L,C_i\cap\cdots\cap C_L\neq\emptyset\}.
	\end{align*}
	Then, their union $\mathcal{C}=\bigcup_{i\in[L]}\mathcal{R}_{\ge i}$ is a congestion-approximator with quality $16\alpha\beta L^2$.
\end{theorem}

The partitions $\mathcal{P}_i$ described in the theorem essentially form a weak expander hierarchy, where each level is essentially a (boundary-linked) weak expander decomposition. Using the existence of routings guaranteed by (2) and (3),~\cite{DBLP:conf/soda/0006R025} show that a demand respecting the congestion-approximator can be iteratively routed. They then show that this weak expander hierarchy can be constructed using existing tools for constructing expander decompositions~\cite{saranurak2019expander,DBLP:conf/soda/0006NPS23}. By doing this, they construct a sequence of partition $\mathcal{P}_1,\ldots,\mathcal{P}_L$ satisfying properties (1), (2), and (3) with parameters $\alpha=O(\log^5{n})$ and $\beta=O(\log^3{n})$.

Our main observation is that we do not need the full power of a weak expander decomposition at each level in order to show the existence of this routing. If a small constant fraction of the vertices (measured in terms of volume in the subgraph) do not have the expander mixing property, this is still sufficient to show that the congestion-approximator routing exists. Specifically, we relax conditions (2) and (3) on the partitions $\mathcal{P}_i$ to allow for a small constant fraction of edges to not participate in the routings at level $i$ and be instead be handled at level $i+1$. 
More formally, we let $\mathcal{P}_1,\ldots,\mathcal{P}_L$ be partitions of subsets $V_1,\ldots,V_L\subseteq V$. We only require the mixing properties (2) and (3) on $V_{i+1}$ for each $i\in[L-1]$, so we should think of $\mathcal{P}_i$ as a weak expander decomposition of $V_i$. To extend the partitions $\mathcal{P}_i$ of $V_i$ to a partition $\overline{\mathcal{P}}_i$ of $V$, we define $\mathcal{Q}_i$ to be the induced partition from the previous level $\overline{\mathcal{P}}_{i-1}$ on $V\setminus V_i$. That is, we define 
    $$Q_i=\{C\cap (V\setminus V_i):C\in\overline{\mathcal{P}}_{i-1},C\cap (V\setminus V_i)\neq\emptyset\}.$$
Then we can define $\overline{\mathcal{P}}_{i}=\mathcal{P}_i\cup\mathcal{Q}_i$. Intuitively, when we only have a partition on $V_i$, we are giving up on routing the demand from the intercluster edges from $\overline{\mathcal{P}}_{i-1}$ in $V\setminus V_i$ and dealing with it at a higher level. In order to move it to the higher level, we include it in $\overline{\mathcal{P}}_i$ through the definition of $\mathcal{Q}_i$.

We now state a morally true version of our relaxed conditions for constructing congestion-approximators.

\begin{theorem}[Informal version of~\Cref{thm:congestion-approximator-guarantee}]
    Consider a capacitated graph $G=(V,E, \bc)$, and let $\alpha\ge 1$ and $\beta\ge1$ be parameters. Let $\mathcal{P}_1,\ldots,\mathcal{P}_L$ be partitions of $V_1,\ldots,V_L$, respectively and extend these to partitions $\overline{\mathcal{P}}_1,\ldots,\overline{\mathcal{P}}_L$ as described above. Suppose the partitions $\overline{\mathcal{P}}_1,\ldots,\overline{\mathcal{P}}_L$ satisfy:
	\begin{enumerate}
		\item $\overline{\mathcal{P}}_1$ is the partition $\{\{v\}:v\in V\}$ of singleton clusters and $\overline{\mathcal{P}}_L$ is the partition $\{\{V\}\}$ with a single cluster.
		\item For each $i\in[L-1]$, for each $C\in\mathcal{P}_{i+1}$, the intercluster edges of $\overline{\mathcal{P}}_{i}$ internal to $C$ along with the boundary edges of $C$ mix in the graph $G$. Moreover, the mixings over all the clusters $C\in\mathcal{P}_{i+1}$ have congestion $\alpha$ simultaneously.
		\item For each $i\in[L-1]$, there is a flow in $G$ with congestion $\beta$ such that each intercluster edge of $\mathcal{P}_{i+1}$ sends its capacity in flow, and each intercluster edge of $\overline{\mathcal{P}}_i$ receives at most a quarter its capacity in flow.
	\end{enumerate}
	For each $i\in[L]$, let partition $\mathcal{R}_{\ge i}$ be the common refinement of partitions $\overline{\mathcal{P}}_i,\overline{\mathcal{P}}_{i+1},\ldots,\overline{\mathcal{P}}_L$, i.e.,
	\begin{align*}
		\mathcal{R}_{\ge i}=\{C_i\cap \cdots\cap C_L:C_i\in\overline{\mathcal{P}}_i,\ldots,C_L\in\overline{\mathcal{P}}_L,C_i\cap\cdots\cap C_L\neq\emptyset\}.
	\end{align*}
	Then, their union $\mathcal{C}=\bigcup_{i\in[L]}\mathcal{R}_{\ge i}$ is a congestion-approximator with quality $48\alpha\beta L^2$.
\end{theorem}

The first advantage of this relaxation is that our algorithm for constructing the partitions $\mathcal{P}_1,\ldots,\mathcal{P}_L$ is faster. In particular, we can use approximate max flow algorithms to implement the cut-matching game. This may cause some nodes to be ``deleted,'' as described in the previous subsection, but this is okay for us since we only need expander mixing guarantees on a (large) constant fraction of the vertices for property (2). In contrast,~\cite{DBLP:conf/soda/0006R025} used a fair-cuts algorithm~\cite{DBLP:conf/soda/0006NPS23} to implement the same step in their paper, which incurred several additional log factors in their runtime.

The second advantage is for obtaining smaller $\beta$. In~\cite{DBLP:conf/soda/0006R025}, they prove that the flow from $\partial\mathcal{P}_{i+1}$ to $\partial\mathcal{P}_{i}$ exists using the boundary-linkedness property of the expander decompositions. This approach naturally suffers from $\beta=\Omega(\log^3{n})$ because the flows guaranteed by the expander decomposition have congestion $\Omega(\log^3{n})$. Instead, our approach is to directly attempt to send flow from $\partial\mathcal{P}_{i+1}$ to $\partial\mathcal{P}_{i}$ at each level. Using a max-flow/min-cut algorithm, we will find a (possibly empty) cut and a flow which saturates the cut. If we simply remove the vertices which are cut out, we have that in the remainder of the graph, there is a flow from $\partial\mathcal{P}_{i+1}$ to $\partial\mathcal{P}_{i}$ (since the flow saturates the cut). This gives us the desired $\beta=O(1)$, and also crucially uses our relaxation of properties (2) and (3).

Finally, we note that we obtain a smaller $\alpha=O(\log^3{n})$ in our construction. This is because we construct a weak expander decomposition on each level, which we observed is sufficient for the simultaneous mixing guarantees required by property (2). This enables us to avoid the costly trimming step used in~\cite{DBLP:conf/soda/0006R025}, and also enables our speedup using warm-starting, which no longer helps for strong expander decompositions.

\paragraph{The approximate max flow algorithm.}
We apply Sherman's framework~\cite{sherman2017area} to convert a congestion-approximator into an approximate max flow algorithm.
Our approach for constructing a congestion-approximator is to construct a weak expander hierarchy using the cut-matching game. In implementing the cut-matching game, we need to solve flow problems (approximately), and we do this using the previous layers of the hierarchy as a ``pseudo''-congestion-approximator. Our (pseudo)-congestion-approximators have quality $\alpha\beta L^2=O(\log^5{n})$, giving us an $O(m\log^6{n} \log \log n)$ time algorithm for solving the flow problems in the cut-matching game. Our improved weak expander decomposition algorithm takes $T=O(\log^2{n})$ rounds to obtain a full weak expander decomposition. Finally, there are $L=O(\log{n})$ layers of the hierarchy, totalling $O(m\log^9{n} \log \log n)$ runtime.

\section{Organization of the Paper}
In~\Cref{sec:weak-expander-decomp} we give a faster algorithm for weak expander decompositions. This algorithm supports implementing its flow subroutines with weaker than usual properties, which are described in~\Cref{ora: matching player} and~\Cref{ora: grafting flow}. In~\Cref{sec:proof-of-congestion-approximator} and~\Cref{sec:bottom-up-construction} we give a faster algorithm for computing a congestion-approximator from the bottom up, using the faster weak expander decomposition as a critical subroutine. To do this, we show how to efficiently implement~\Cref{ora: matching player} and~\Cref{ora: grafting flow}. As a direct consequence of our new algorithm for constructing a congestion-approximator, we obtain the fastest known approximate max flow algorithm. In~\Cref{sec:fair}, we give a faster algorithm for one-sided fair cuts which we need to implement~\Cref{ora: grafting flow}.~\Cref{sec: omit} includes all other technical proofs omitted from the main body of the paper.

\section{Preliminaries} 
\paragraph{Sets.} We use $\Z_{\geq 0}$ to denote the set of non-negative integers. For $k$ a positive integer, we use $[k]$ to denote $\{1,2, \ldots, k\}$. When $S \subseteq X$, we sometimes use $\overline{S}$ to denote its complement; that is, $\overline{S} = X \setminus S$.

\paragraph{Functions.} For two functions $f, g : X \to \R$ let $f \leq g$ denote that, for all $x \in X$, $f(x) \leq g(x)$. We write $\supp(f)$ to denote the subset of $X$ on which $f$ takes nonzero values. For $S \subseteq X$ and $S$ finite, we also write $f(S)$ as shorthand for $\sum_{x \in S} f(x)$. We use $f|_{S}$ to mean the restriction of $f$ to $S$. When clear from context, we sometimes abuse notation and use $f|_{S}$ to denote the function $f$ on the same domain but set equal to $0$ outside of $S$.  We also use all of the above notation for vectors, interpreting those vectors as functions. We often bold vectors to distinguish them from scalars (e.g., write $\bd \in \R^V$). 

\paragraph{Graphs.} We consider capacitated (weighted) graphs $G = (V, E, \bc)$ where $\bc \in [1,W] \cap \Z$. Unless otherwise specified, we use $n$ to denote the order of $G$ and $m$ to denote its size. Sometimes we write $V_G$ (or $V(G)$), $E_G$ (or $E(G)$), and $\bc_G$ to clarify that they are the parameters of the graph $G$. For $S \subseteq V$, denote the induced subgraph of $G$ as $G[S]$. In other words, $G[S]$ is the subgraph of $G$ formed by retaining exactly vertices in $S$ and edges between vertices in $S$. 

Given a partition $\calA$ of $V$, we write $\partial \calA$ to denote the set of intercluster edges in $G$. When $\calA$ is just a single cut $(S, V \setminus S)$, we sometimes write $\partial S$ instead. We often consider $\partial \calA$ as an edge subgraph of $G$. We also use the notation $\delta \calA = \bc(\partial \calA)$ and $\delta S = \bc(\partial S)$ (or $\delta(\calA)$ and $\delta(S)$) as shorthand denoting the total capacity of intercluster or cut edges.  For $u \in V$, we denote the (weighted) degree of $u$ in $G$ as $\deg_G(u) = \delta_G(\{u\})$. We will often also consider $\deg_{\partial \calA}(u) = \delta_{\partial \calA}(\{u\})$ which treats $\partial \calA$ as an edge subgraph of $G$. Finally, throughout we consider vertex weights $\bd \in \Z_{\geq 0}^V$, with the most common weight function being $\bd = \deg_G$ or $\deg_H$ for $H$ a subgraph of $G$. 

\paragraph{Flow.}
A \textit{demand} is a vector $\bb \in \R^V$ whose entries sum to $0$. We say a flow $f : E \to \R$ routes a demand $\bb$ if for each $v \in V$ the net flow at $v$ in $f$ is $\bb(v)$. We say that $f$ has \textit{congestion} $\kappa$ if the flow through any edge in $f$ is at most $\kappa$ times its capacity. Given a flow $f$, a  \textit{path decomposition} of $f$ is a collection of weighted paths in $G$ such that, for each $(u,v) \in E$, the flow from $u$ to $v$ in $f$ is the sum of weights of paths containing the edge from $u$ to $v$ in the path decomposition.

\paragraph{Expansion.} Let $G = (V,E, \bc)$ be a capacitated graph, and let $\bd \in \R_{\geq 0}^V$ be a vertex weighting. Let $S \subseteq V$. Then, the \textit{conductance of $S$ in $G$ with respect to $\bd$} is
\[
\Phi_{G, \bd}(S) = \frac{\delta_G(S)}{\min(\bd(S), \bd(V \setminus S))}.
\]
 We say that a cut $S$ is \textit{$\phi$-sparse} (in $G$ with respect to $\bd$) if $\Phi_{G, \bd}(S) \leq \phi$. We say that $G$ is a $(\phi, \bd)$-\textit{expander} if, for all $S \subseteq V$, we have $\Phi_{G, \bd}(S) \geq \phi$. For $A \subseteq V$, we say that $A$ is \textit{$\phi$-nearly $\bd$-expanding} in $G$ (or $A$ is a $(\phi, \bd)$-\textit{near-expander} in $G$) if, for all $S \subseteq A$, we have 
\[
\frac{\delta_G(S)}{\min(\bd(S), \bd(A \setminus S))} \geq \phi.
\]
Note that if $A$ is $\phi$-nearly $\bd$-expanding in $G$, then the same holds for all $A' \subseteq A$, since the denominators of the relevant expressions only decrease. When $\bd = \deg_G$ or $\bd$ is clear from context we say $G$ is a $\phi$-expander (respectively, near-expander). 

We can also define expansion with respect to flows. We say that a vertex weighting $\bd \in \R_{\geq 0}^V$ \textit{mixes in} $G$ \textit{with congestion} $\kappa$ if, for all demands $\bb \in \R^V$ with $|\bb| \leq \bd$, we have that $\bb$ is routable in $G$ with congestion at most $\kappa$. In fact, $\bd$ mixes in $G$ with congestion $1/\phi$ if and only if $G$ is a $(\phi, \bd)$-expander. Note that while $\bd|_A$ mixing in $G$ with congestion $\kappa$ implies that $A$ is a $(\phi, \bd)$-near-expander in $G$, the converse does not hold in general. Flow-based expansion is stronger than cut-based expansion for near-expanders.

Sometimes we require an even stronger notion of expansion with respect to multi-commodity flows. We say that a collection of vertex weights $\{\bd_i \,:\, i \in I\}$ \textit{mixes simultaneously} in $G$ with congestion $\kappa$ if, for all tuples of demands $(\bb_i)_{i \in I}$ with each $\bb_i \in \R^V$ satisfying $|\bb_i| \leq \bd_i$, there exists a multicommodity flow $\bF$ with one commodity per demand which routes all $\bb_i$ and has total congestion $\kappa$. 

\paragraph{Congestion-approximators.} 
Given a graph $G = (V,E, \bc)$, a \textit{congestion-approximator} $\calC$ of \textit{quality} $\alpha$ is a family of subsets of $V$ such that, for any demand $\bb$ satisfying $|\bb(C)| \leq \delta_G(C)$ for all $C \in \calC$, there is a flow routing demand $\bb$ with congestion $\alpha$.

\section{Faster Algorithm for Weak Expander Decomposition}\label{sec:weak-expander-decomp}
 Our input is an undirected, capacitated graph $G= (V,E, \bc)$ of order $n$, size $m$, and with $\bc \in \Z \cap [1,W]$; an expansion parameter $\phi > 0$; and a vertex weighting $\bd : V \to \mathbb{Z}_{\geq 0}$. For intuition, it may be helpful to think of $\bd$ as $\deg_G$. Our goal is to compute a decomposition of $V$,  $\calA_T$, and some $\bd_T \leq \bd$ such that:
\begin{enumerate}
    \item \textit{Decomposition into expanders:} $\{\bd_T |_A \,:\, A \in \calA_T\}$ mix simultaneously in $G$ with congestion $O(1/\phi)$.
    \item \textit{Few cut edges:} The total capacity of edges cut by $\calA_T$ is $O(\phi \bd(V) \log nW)$.
    \item \textit{Limited deleted demand:} $\bd(V) - \bd_T(V) \leq \deleps \bd(V)$, for some small constant $\deleps > 0$.
\end{enumerate}
In the case of implementing the matching steps of cut-matching with an exact max flow oracle or a fair cuts-based approximate max flow oracle, we can set $\deleps = 0$. For our application to approximate max flow, our max flow oracles are too weak to achieve such a guarantee $\deleps = 0$, but the guarantee of (3) still suffices. We defer stating the main result of this section,~\Cref{thm: weak-decomp-del2}, so it can be stated in the context of the flow oracles it assumes.

\subsection{Weak Expander Decomposition with Deleted Demand} \label{sec: weak expander wdd}

We begin by stating some basic definitions. For $A \subseteq V$ and a vertex weighting $\bb : A \to \mathbb{Z}_{\geq 0}$, define an $A$-\textit{commodity flow} as a multicommodity flow where each $v \in A$ is a source of $\bb(v)$ of its unique flow commodity.  

For the purposes of our analysis, we will implicitly maintain a \textit{flow matrix} $\bF \in \R^{V \times V}_{\geq 0}$ throughout. We say a flow matrix $\bF$ is \textit{routable with congestion $\kappa$} if there exists a $V$-commodity flow $f$ such that for each $(u,v) \in V^2$, $f$ simultaneously routes $\bF(u,v)$ of $u$'s commodity to $v$ with no edge $e$ having more than $\kappa \bc(e)$ flow passing through it.

We initialize our flow matrix as $\bF_0 = \diag(\bd)$, where we view $\bd$ as a vector in $V^{\mathbb{Z}_{\geq 0}}$. We also initialize our set of ``deleted vertices'' as the empty set; $D_0 = \emptyset$. At each step $t \geq 1$, $\bd_t$ is equal to $\bd$, except set equal to $0$ on $D_{t-1}$. The algorithm then proceeds in $T$ rounds. 

Some components of $\calA_t$ become \textit{inactive} over the course of the algorithm. The components that are \textit{active} are those for which we have not marked all nodes as deleted and for which we have not certified that the component is an expander. We decompose $\calA_t = \calA_t^\circ \sqcup  \calA_t^\times$ into the active and inactive portions, respectively. In each step, we will make progress towards certifying expansion on some $\calA'_t \subseteq \calA_t^\circ$. For each active component $A \in \calA_t^\circ$, we maintain a counter $x^A_t$ recording how many times we have made progress on this set. When the counter is high enough, we have certified expansion on $A$ with high probability and we can set the component to be inactive. We state the algorithm more formally in~\Cref{alg:weak-ed-del}.

\begin{algorithm}[t]
	\caption{Weak Expander Decomposition with Deletions}
	\label{alg:weak-ed-del}
	\fbox{
		\parbox{0.97\columnwidth}{
		  \textsc{WeakDecomp}$(G=(V,E,\bc), \bd, \phi)$ \\
            \tab Initialize $\mathcal{A}_0=\{V\}$, $x^V_0 = 0$, $\bd_0 = \bd$, and implicitly initialize $\bF_0 = \diag(\bd)$\\
            \tab \textbf{for} $t\in[T]$: \\
            \tab \tab \textbf{for} $A\in\mathcal{A}_{t-1}^\times$: add $A$ to $\mathcal{A}_t^\times$ and set $x_t^A = x_{t-1}^A$;\\
            \tab \tab \textbf{for} $A\in\mathcal{A}_{t-1}^\circ$: \hfill \textcolor{gray}{// Cut step}\\
            \tab \tab \tab Find sets $L_A, R_A \subseteq A^\circ$, where $A^\circ = \{v \in A \,:\, \bd_{t-1}(v) > 0\}$\\
            \tab \tab Compute cuts $C_A$ and matching $\bM_t$ between $L_A$ to $R_A$ in each $A \in \mathcal{A}_{t-1}^\circ$ \hfill \textcolor{gray}{// Matching step}\\
            \tab \tab \textbf{for} $u \in V$: \hfill \textcolor{gray}{// Update $\bd_t$}\\
            \tab \tab \tab \textbf{if} $u \in L_A \cap (A \setminus C_A) \subseteq A$ for $A \in \calA_{t-1}' \subseteq \calA_{t-1}^\circ$: \\
            \tab \tab \tab \tab Let $\bM_t(u) = \sum_{v \in A^\circ} \bM_t(u,v)$\\
            \tab \tab \tab \tab \textbf{if} $\bM_t(u) < \bd(u)/2$: set $\bd_t(u) = 0$; \\
            \tab \tab \tab \textbf{else:} $\bd_t(u) = \bd_{t-1}(u)$; \\
             \tab \tab \textbf{for} $A\in\mathcal{A}_{t-1}^\circ$: \hfill \textcolor{gray}{// Form $\calA_t$, update counters}\\
            \tab \tab \tab \textbf{for} $S \in \{C_A, A\setminus C_A\}$ nonempty: \\
             \tab \tab \tab \tab \textbf{if} $A \in \mathcal{A}_{t-1}'$: set $x_t^{S} = x_{t-1}^A + 1$; \\
            \tab \tab \tab \tab \textbf{else}: set $x_t^{S} = x_{t-1}^A$; \\
            \tab \tab \tab\tab \textbf{if} $\bd_t(S) \leq 15\bd(S)/16$: \\
            \tab \tab \tab \tab \tab  \textbf{for} $u \in S$: set $\bd_t(u) = 0$; \\
            \tab \tab \tab \tab \textbf{if} $\bd_t(S) = 0$ or $x_t^S > 10^5C \log n \log nW$ or $|\supp(\bd_t|_S)| = 1$: add $S$ to $\calA_t^\times$; \hfill \textcolor{gray}{// $C$ sufficiently large constant}\\
            \tab \tab \tab \tab \textbf{else:} add $S$ to $\calA_t^\circ$; \\
            \tab \tab \textbf{for} $u \in V$: \hfill \textcolor{gray}{// Update $\bF_t$}\\
            \tab \tab \tab \textbf{if} $u \in A \in \calA_{t-1}^\circ$: \\
            \tab \tab \tab \tab Set $\bF_t(u) = \left( 1 - \frac{\bM_t(u)}{2\bd(u)}\right) \bF_{t-1}(u) + \frac{1}{2} \sum_{w \in A^\circ} \frac{\bM_t(u,w)}{\bd(w)} \cdot \bF_t(w)$ \\
            \tab \tab \tab \textbf{else:} set $\bF_t(u) = \bF_{t-1}(u)$; \\
            \tab \textbf{return} $\calA_T$
        }
        }
\end{algorithm}

\paragraph{Cut Player.} We implement our cut step as follows.  

\begin{itemize}
    \item Sample $\br_t \in \R^V$, a random unit vector.
    \item For each $A \in \calA_{t-1}^\circ$:
    \begin{itemize}
        \item For each $u \in A^\circ = \{v \in A \,:\, \bd_{t-1}(v) > 0\}$, compute 
    \[
    p_t(u) = \inner{\frac{\bF_{t-1}(u)}{\bd(u)}}{\br_t}.
    \]
    \item Compute a partition $L_A \sqcup R_A = A^\circ$ such that:
    \begin{enumerate}
        \item $\bd_{t-1}(L_A) = \left \lceil \bd_{t-1}(A)/8 \right \rceil$.
        \item $\max_{u \in L_A} p_t(u) \leq \eta \leq \min_{u \in R_A} p_t(u)$ or $\max_{u \in R_A} p_t(u) \leq \eta \leq \min_{u \in L_A}$.
    \end{enumerate}
    \end{itemize}
\end{itemize}
In particular, we compute a partition of $A^\circ$ such that $L_A$ is guaranteed to contain a set of vertices certifying progress towards expansion. We show the existence of such a set certifying progress via a technical lemma. This is a variant of Lemma 5.15 of~\cite{directed-expander-decomposition} adapted to the undirected setting.

\begin{restatable}{lemma}{progset}\label{lem: progress set}
  Let $X$ be a finite multi-subset of $\R$ with $|X| \geq 2$.  There exists $\eta \in \R$ inducing a partition $X = L_{\eta} \sqcup R_{\eta}$ with $\max(L_{\eta}) \leq \eta \leq \min(R_{\eta})$ or $\max(R_{\eta}) \leq \eta \leq \min(L_{\eta})$, $|L_\eta| = \left \lceil |X|/8 \right \rceil$, and with the following additional guarantees.  Define $\bar{\mu} = \frac{1}{|X|}\sum_{x \in X} x$. There exists $S \subseteq L_{\eta}$ such that
    \begin{enumerate}
        \item For each $s\in S$, we have $(s-\eta)^2\ge \frac{1}{9}\cdot (s-\bar{\mu})^2$.
        \item $\sum_{s\in S}(s-\bar{\mu})^2\ge \frac{1}{36}\sum_{x\in X}(x-\bar{\mu})^2$.
    \end{enumerate}   
\end{restatable}
The proof is similar to the proof of Lemma 5.15 in~\cite{directed-expander-decomposition}. We defer the proof to the appendix for the sake of exposition. To construct our partition $L_A \sqcup R_A$, we apply~\Cref{lem: progress set} to the multiset of $p_t(u)$, where each value is repeated $\bd_t(u)$ times. Elements whose duplicates appear in both $L_\eta$ and $R_\eta$ have $p_t(u) = \eta$ and will not be relevant for our potential reduction analysis. 

\begin{lemma} \label{lem: cut step runtime}
For each $t \leq T$, we can compute the cut step partitions for all $A \in \calA_{t-1}^\circ$ in total time $O(mt)$.
\end{lemma}
\begin{proof}
To compute $p_t(u)$ for each $u \in A^\circ \in \calA_{t-1}^\circ$, we use the recursive formulation of $\bF_{t-1}(u)$ (see~\cref{eqtn: recursive F}). For $t = 1$, we just use the initialization. Otherwise, since $\bM_{t-1}$ has at most $m$ nonzero entries (as do all the other computed matching matrices), we can compute $\inner{\bF_{t-1}(u)}{ \br_t}$ for all such  $u \in A^\circ \in \calA_{t-1}^\circ$ given each $\inner{\bF_{t-2}(u)}{ \br_t}$ in total time $O(m)$. The base case inner products take $O(m)$ time to compute, so in total it takes $O(mT)$ time to compute all $p_t(u)$ for $t \leq T$. 

Given the $p_t(u)$, for each $A \in \calA_{t-1}^\circ$, we need to decide whether to choose $\eta$ in the application of~\Cref{lem: progress set} so that $r|X|$ or $(1-r)|X|$ of the $p_t(u)$ are at most $\eta$, where $r= \frac{\left \lceil |X|/8| \right \rceil}{|X|}$. This takes $O(|A|)$ time (we can, for example, check whether~\Cref{eqtn: case 1} holds). Since $\calA_{t-1}$ is a partition, this only contributes an additional $O(n)$ total time.
\end{proof}

\paragraph{Matching Player.}

In the $t$\textsuperscript{th} matching step, we consider the following flow problem. Let $G_t$ be the graph with all edges between different components in $\calA_{t - 1}$ deleted and all remaining edges with capacity scaled by $2/\phi$. For $v \in L_A$ for some $A \in \calA_{t-1}^\circ$, we set its source to be $\Delta(v) = \bd_{t-1}(v)$. For $v \in R_A$, we set its sink to be $\nabla(v) = \bd_{t-1}(v)$. 

We assume access to an approximate max flow oracle with the following guarantees.

\begin{oracle}[Matching Player Flow Oracle] \label{ora: matching player}
On such a flow instance, we find $\calA'_{t-1} \subseteq \calA_{t-1}^\circ$ with 
\[
 \bd_{t-1}\left(\bigcup_{A \in \calA'_{t-1}} A\right) \geq \frac{1}{2} \bd_{t-1}\left(\bigcup_{A \in \calA_{t-1}^\circ} A\right)
\]
such that, for each $A \in \calA'_{t-1}$:
\begin{enumerate}
    \item  We find a (possibly empty) cut $C_A \subseteq A$ with $\bd_{t-1}(C_A) \leq \bd_{t-1}(A)/2$. In addition, we have that the total capacity of computed cuts $(C_A, A \setminus C_A)$ is at most
    \[
    \frac{\phi}{8} \sum_{A \in \calA'_{t-1}} \bd_{t-1}(C_A) + 2\floweps \bd(V).
    \]
    \item We find a flow routing at least $\Delta(A \setminus C_A) - 2\floweps \bd(A)$ source of $\Delta|_{A \setminus C_{A}}$ within $A$, for $\floweps < 1/2$.
\end{enumerate}

\end{oracle}

\begin{remark}
Note that, if given access to a $(1-\floweps)$-approximate max flow oracle, we could apply the oracle on each $A \in \calA_{t-1}^\circ$ and get the desired properties with $\calA_{t-1}' = \calA_{t-1}^\circ$. 
\end{remark}

\paragraph{Updating the flow matrix.} Each application of the flow oracle induces a weighted matching in each component between the source and the sink. We update the implicit flow matrix $\bF$ accordingly. In particular, let $\bM_t \in V \times V$ be the symmetric matrix where, for $u$ a source vertex, $\bM_t(u,v)$ is the amount of flow sent from vertex $u$ to vertex $v$ in the flow (after computing some path decomposition using link-cut trees). Importantly, since $\bM_t(u,v)$ is formed by a path decomposition, we can guarantee that $\bM_t(u,v)$ has at most $m$ nonzero entries. For convenience, for $u \in A$, define $\bM_t(u) := \sum_{v \in A} \bM_t(u,v)$. Also define $A^\circ := \{v \in A \,:\, \bd_{t-1}(v) > 0\}$. We can then define $\bF_t$ recursively from $\bF_{t-1}$ and $\bM_t$ as follows. For $u \in A \in \calA_{t-1}^\circ$,
\begin{equation}
    \bF_t(u) = \left( 1 - \frac{\bM_t(u)}{2\bd(u)}\right) \bF_{t-1}(u) + \frac{1}{2} \sum_{w \in A^\circ} \frac{\bM_t(u,w)}{\bd(w)} \cdot \bF_{t-1}(w).  \label{eqtn: recursive F}
\end{equation}
For $u \in A \in \calA_{t-1}^\times$, $u$ is not involved in the matching step, so $\bF_t(u) = \bF_{t-1}(u)$. 
\begin{claim} \label{cla: route d(v) flow}
For all $t$ and for all $u \in V$ with $\bd_t(u) > 0$, 
\[
\sum_{w \in V} \bF(u,w) = \bd(u).
\]
\end{claim}
\begin{proof}
The claim is by definition for $\bF_0$. Then, it remains true inductively as well since 
\[
\left( 1 - \frac{\bM_t(u)}{2\bd(u)}\right) \bd(u) + \frac{1}{2} \sum_{w \in A^\circ} \bM_t(u,w) = \bd(u)  - \frac{\bM_t(u)}{2} + \frac{\bM_t(u)}{2}.
\]    
\end{proof}
\begin{claim} \label{cla: congestion bd}
For all $t \geq 0$, $\bF_t$ is routable with congestion $2t/\phi$.
\end{claim}
\begin{proof}
Note that we can route the flow inducing each $\bM_t$ with congestion $2/\phi$ by definition of the flow instance. Then, for each flow update step, if $u$ and $w$ have $\bM_t(u,w) > 0$, for each $v \in V$, we route $\frac{\bM_t(u,w)}{2 \bd(w)} \bF_t(w,v)$ of commodity $u$ to $v$ by routing from $u$ to $w$ and then from $w$ to $v$. The only new congestion arises from the flow between the matched vertices, yielding the desired result.
\end{proof}

\paragraph{Convergence Analysis.}
To prove that we achieve the desired decomposition after $T = O(\log n \log nW)$ total steps, we consider the following potential function. For each $A \in \calA_{t}^\circ$, first define the average flow vector 
\[
\mu_t^A = \frac{1}{\bd_t(A)} \cdot \sum_{u \in A^\circ} \bF_t(u).
\]
As before, $A^\circ := \{u \in A \,:\, \bd_t(u) > 0\}$. Then, 
\[
\psi_t(A) = \bd_t(A) \cdot \sum_{u \in A^\circ} \bd(u) \norm{\frac{\bF_t(u)}{\bd(u)}-\mu_t^A}_2^2.
\]
First we show that small potential implies simultaneous mixing.
\begin{lemma} \label{lem: small potential mixing}
Suppose that $\psi_t(A) < 1/(mW)^8$ for all $A \in \calA_t$ with $A^\circ \neq \emptyset$. Then, 
\[
\{ \bd_t |_A \,:\, A \in \calA_t\} 
\]
mix simultaneously in $G$ with congestion $4t/\phi$.
\end{lemma}
\begin{proof}
Let $\calA_t = \{A_1, A_2, \ldots, A_r\}$. For each $i \in [r]$, let $\bb_i$ be a demand respecting $\bd_t |_{A_i}$. That is, $|\bb_i| \leq \bd_t|_{A_i}$ and $\bb_i(A_i) = 0$.

Now, consider the multicommodity flow with one commodity per demand that, for each $i \in [r]$ and each $u \in A_i^\circ$, sends $\bb_i(u) \cdot \frac{\bF_t(u,w)}{\bd(u)}$ to each $w \in V$. Since we have 
\[
\sum_{w \in V} \bF_t(u,w) = \bd(u), 
\]
by~\Cref{cla: route d(v) flow}, we have that this multicommodity flow routes $\bb_i(u)$ flow out of $u$ in the $i$\textsuperscript{th} commodity. Since for $j \neq i$, $\bb_j(u) = 0$, the net flow out in the other commodities is $0$. It then remains to show that the net flow into $u$ is approximately $0$ in every commodity. 

To this end, for any $j \in [r]$, the net flow that $u$ receives is
\begin{align*}
\left| \sum_{w \in A_j^\circ} \bb_j(w) \frac{\bF_t(w,u)}{\bd(w)} \right|&=  \left|\sum_{w \in A_j^\circ}\bb_j(w) \left(\frac{\bF_t(w,u)}{\bd(w)} - \mu_t^{A_j} + \mu_t^{A_j} \right)  \right| \\
 &= \left| \sum_{w \in A_j^\circ}\bb_j(w) \left(\frac{\bF_t(w,u)}{\bd(w)} - \mu_t^{A_j} \right)  \right|\tag{Since $\bb_j(A_j) = 0$} \\
 &\leq 
 \sum_{w \in A_j^\circ}|\bb_j(w)| \left|\frac{\bF_t(w,u)}{\bd(w)} - \mu_t^{A_j} \right| \\
 &\leq \sum_{w \in A_j^\circ} \frac{\bd(w)}{(mW)^4}  \tag{$|\bb_j| \leq \bd_t|_{A_j} \leq \bd$ and $\psi_t(A_j) \leq 1/(mW)^8$} \\
 &\leq \frac{1}{(mW)^3}.
\end{align*}
So, $u$ receives $\bb_j(u) \pm \frac{1}{(mW)^3}$ in commodity $j$ via a routing of congestion $2t/\phi$ (using the congestion bound from~\Cref{cla: congestion bd}). We can then trivially route the residual error demand with congestion at most $n^3 \cdot \frac{1}{(mW)^3} \leq 1/\phi$. The total congestion is then at most $2t/\phi + 1/\phi \leq 4t/\phi$, as needed.
\end{proof}

\begin{remark} \label{rmk:-fast-mixing}
Given demands $\bb_i$ respecting $\bd_t |_{A_i}$ for each $i \in [r]$, we can compute the mixing routing in $O( mt \log nW)$ time by rescaling the paths in the path decompositions computed in the prior rounds' matching steps. Like the original path decomposition computation, we can implement this using link-cut trees~\cite{DBLP:journals/jcss/SleatorT83}.
\end{remark}
Now that we have shown $\psi_t$ is a good potential function, we show that we decrease $\psi_t$ in each round. 
\begin{lemma} \label{lem: progress lem}
 Let $t \geq 0$. For each $A \in \calA_{t}' \subseteq  \calA_{t}^\circ$ and $S \subseteq A$ with $\bd_{t+1}(S) \neq 0$, we have
 \begin{align*}
 \frac{\psi_t(A)}{\bd_t(A)} - \frac{\psi_{t+1}(S)}{\bd_{t+1}(S)} &\geq  \frac{1}{2} \sum_{u \in S^\circ} \bM_t(u)\left(1 - \frac{\bM_t(u)}{2 \bd(u)} \right) \norm{\frac{\bF_t(u)}{\bd(u)} - \sum_{w \in A^\circ} \frac{\bM_t(u,w)}{\bM_t(u)} \frac{\bF_t(w)}{\bd(w)}}_2^2 \\
 &+ \frac{1}{2} \sum_{u \in A^\circ \setminus S^\circ} \bd(u) \norm{\frac{\bF_t(u)}{\bd(u)} - \mu_t^A}_2^2.
 \end{align*}
\end{lemma}
\begin{proof}
First, note that
\begin{equation} \label{eqtn: bd1}
  \frac{\psi_{t+1}(S)}{\bd_{t+1}(S)} = \sum_{u \in S^\circ} \bd(u)   \norm{\frac{\bF_{t+1}(u)}{\bd(u)} - \mu_{t+1}^S}_2^2 \leq \sum_{u \in S^\circ} \bd(u)   \norm{\frac{\bF_{t+1}(u)}{\bd(u)} - \mu_{t}^A}_2^2.
\end{equation}
This uses the fact that
\[
\mu_{t+1}^S = \argmin_{\bx \in \R^V} \sum_{u \in S^\circ} \bd(u)   \norm{\frac{\bF_{t+1}(u)}{\bd(u)} - \bx}_2^2,
\]
which in turn can be verified by computing the derivative coordinate-wise and setting it equal to $0$. 

We further manipulate this upper bound on $  \frac{\psi_t(A)}{\bd_t(A)} - \frac{\psi_{t+1}(S)}{\bd_{t+1}(S)} $, expanding the recursive definition of $\bF_{t+1}(u)$. Continuing from~\Cref{eqtn: bd1}, we can upper bound $\frac{\psi_t(A)}{\bd_t(A)} - \frac{\psi_{t+1}(S)}{\bd_{t+1}(S)} $ by

{
\allowdisplaybreaks
\begin{align*}
&\sum_{u \in S^\circ} \bd(u)   \norm{\frac{1}{\bd(u)} \cdot \left(\left( 1 - \frac{\bM_t(u)}{2\bd(u)}\right) \bF_{t-1}(u) + \frac{1}{2} \sum_{w \in A^\circ} \frac{\bM_t(u,w)}{\bd(w)} \cdot \bF_t(w)\right) - \mu_{t}^A}_2^2 
 \\
 &= \sum_{u \in S^\circ} \bd(u)   \norm{ \left( 1 - \frac{\bM_t(u)}{2\bd(u)}\right) \left( \frac{\bF_{t}(u)}{\bd(u)} - \mu_t^A \right) + \frac{\bM_t(u)}{2\bd(u)} \left(  \sum_{w \in A^\circ} \frac{\bM_t(u,w)}{\bM_t(u)\bd(w)} \cdot \bF_t(w) - \mu_{t}^A\right)}_2^2 \\
 &= \sum_{u \in S^\circ} \bd(u)   \left( 1 - \frac{\bM_t(u)}{2\bd(u)}\right)^2  \norm{\frac{\bF_{t}(u)}{\bd(u)} - \mu_t^A}_2^2 +  \left(\frac{\bM_t(u)}{2\bd(u)}\right)^2 \norm{\sum_{w \in A^\circ} \frac{\bM_t(u,w)}{\bM_t(u)\bd(w)} \cdot \bF_t(w) - \mu_{t}^A}_2^2 \\
 &\quad+ \sum_{u \in S^\circ} \bd(u) \left( 1 - \frac{\bM_t(u)}{2\bd(u)}\right) \left(\frac{\bM_t(u)}{\bd(u)}\right) \inner{\frac{\bF_{t}(u)}{\bd(u)} - \mu_t^A}{\sum_{w \in A^\circ} \frac{\bM_t(u,w)}{\bM_t(u)\bd(w)} \cdot \bF_t(w) - \mu_{t}^A } \\
 &= -\sum_{u \in S^\circ }\bd(u) \left( 1 - \frac{\bM_t(u)}{2\bd(u)}\right) \left(\frac{\bM_t(u)}{2\bd(u)}\right) \norm{\frac{\bF_{t}(u)}{\bd(u)} - \sum_{w \in A^\circ} \frac{\bM_t(u,w)}{\bM_t(u)\bd(w)} \cdot \bF_t(w)}_2^2 \\
 &+\quad \sum_{u \in S^\circ} \bd(u)   \left(\left( 1 - \frac{\bM_t(u)}{2\bd(u)}\right)^2 + \left( 1 - \frac{\bM_t(u)}{2\bd(u)}\right) \left(\frac{\bM_t(u)}{2\bd(u)}\right) \right)  \norm{\frac{\bF_{t}(u)}{\bd(u)} - \mu_t^A}_2^2 \\
 &+\quad \sum_{u \in S^\circ} \bd(u) \left( \left(\frac{\bM_t(u)}{2\bd(u)}\right)^2 + \left( 1 - \frac{\bM_t(u)}{2\bd(u)}\right) \left(\frac{\bM_t(u)}{2\bd(u)}\right)\right)   \norm{\sum_{w \in A^\circ} \frac{\bM_t(u,w)}{\bM_t(u)\bd(w)} \cdot \bF_t(w) - \mu_{t}^A}_2^2 \\
  &= -\sum_{u \in S^\circ }\bd(u) \left( 1 - \frac{\bM_t(u)}{2\bd(u)}\right) \left(\frac{\bM_t(u)}{2\bd(u)}\right) \norm{\frac{\bF_{t}(u)}{\bd(u)} - \sum_{w \in A^\circ} \frac{\bM_t(u,w)}{\bM_t(u)\bd(w)} \cdot \bF_t(w)}_2^2 \\
 &+\quad \sum_{u \in S^\circ} \bd(u)   \left( 1 - \frac{\bM_t(u)}{2\bd(u)}\right) \norm{\frac{\bF_{t}(u)}{\bd(u)} - \mu_t^A}_2^2 + \sum_{u \in S^\circ} \frac{\bM_t(u)}{2}   \norm{\sum_{w \in A^\circ} \frac{\bM_t(u,w)}{\bM_t(u)\bd(w)} \cdot \bF_t(w) - \mu_{t}^A}_2^2 \\
 &\leq -\sum_{u \in S^\circ }\bd(u) \left( 1 - \frac{\bM_t(u)}{2\bd(u)}\right) \left(\frac{\bM_t(u)}{2\bd(u)}\right) \norm{\frac{\bF_{t}(u)}{\bd(u)} - \sum_{w \in A^\circ} \frac{\bM_t(u,w)}{\bM_t(u)\bd(w)} \cdot \bF_t(w)}_2^2 \\
 &+\quad \sum_{u \in S^\circ} \bd(u)   \left( 1 - \frac{\bM_t(u)}{2\bd(u)}\right) \norm{\frac{\bF_{t}(u)}{\bd(u)} - \mu_t^A}_2^2 + \sum_{u \in S^\circ} \sum_{w \in A^\circ} \frac{\bM_t(u, w)}{2}   \norm{\frac{\bF_t(w)}{\bd(w)} - \mu_{t}^A}_2^2 \\
 &\leq -\sum_{u \in S^\circ }\bd(u) \left( 1 - \frac{\bM_t(u)}{2\bd(u)}\right) \left(\frac{\bM_t(u)}{2\bd(u)}\right) \norm{\frac{\bF_{t}(u)}{\bd(u)} - \sum_{w \in A^\circ} \frac{\bM_t(u,w)}{\bM_t(u)\bd(w)} \cdot \bF_t(w)}_2^2 \\
 &+\quad \sum_{u \in S^\circ} \bd(u)   \left( 1 - \frac{\bM_t(u)}{2\bd(u)}\right) \norm{\frac{\bF_{t}(u)}{\bd(u)} - \mu_t^A}_2^2 + \sum_{u \in A^\circ} \frac{\bM_t(u)}{2}   \norm{\frac{\bF_t(u)}{\bd(u)} - \mu_{t}^A}_2^2.
\end{align*}
}
In the first three steps, we just expand the quadratic term. In the next step, we use the expansion of $\norm{\frac{\bF_{t}(u)}{\bd(u)} - \sum_{w \in A^\circ} \frac{\bM_t(u,w)}{\bM_t(u)\bd(w)} \cdot \bF_t(w)}_2^2$ to cancel out the inner product term. In the next two steps, we simplify and then apply Jensen's inequality to the last term, using that $\sum_{w \in A^\circ} \bM_t(u,w) = \bM_t(u)$. Finally, the last inequality follows from swapping the order of summation in the last term and using that $\sum_{u \in S^\circ} \bM_t(u,w) \leq \bM_t(w)$.

Using the above and expanding the definition of $\frac{\psi_t(A)}{\bd_t(A)}$, we then have
\begin{align*}
\frac{\psi_t(A)}{\bd_t(A)} - \frac{\psi_{t+1}(S)}{\bd_{t+1}(S)} &\geq \sum_{u \in S^\circ } \left(\frac{\bM_t(u)}{2}\right) \left( 1 - \frac{\bM_t(u)}{2\bd(u)}\right) \norm{\frac{\bF_{t}(u)}{\bd(u)} - \sum_{w \in A^\circ} \frac{\bM_t(u,w)}{\bM_t(u)\bd(w)} \cdot \bF_t(w)}_2^2 \\
&+ \sum_{u \in A^\circ \setminus S^\circ} \left(\bd(u) - \frac{\bM_t(u)}{2}\right) \norm{ \frac{\bF_t(u)}{\bd(u)} - \mu_t^A}_2^2 \\
&\geq \sum_{u \in S^\circ } \left(\frac{\bM_t(u)}{2}\right) \left( 1 - \frac{\bM_t(u)}{2\bd(u)}\right) \norm{\frac{\bF_{t}(u)}{\bd(u)} - \sum_{w \in A^\circ} \frac{\bM_t(u,w)}{\bM_t(u)\bd(w)} \cdot \bF_t(w)}_2^2 \\
&+ \sum_{u \in A^\circ \setminus S^\circ} \left(\frac{\bd(u)}{2}\right) \norm{ \frac{\bF_t(u)}{\bd(u)} - \mu_t^A}_2^2.
\end{align*}
The last inequality uses $\bM_t(u) \leq \bd(u)$. This yields the desired result.
\end{proof}

We can now use~\Cref{lem: progress lem} to show that we make progress on each $A \in \calA_t'$. We know that $\bd_{t+1}(C_A) \leq \bd_t(A)/2$ from~\Cref{ora: matching player}. Then by~\Cref{lem: progress lem}, setting $S = C_A$,
\[
\frac{\psi_t(A)}{\bd_t(A)} - \frac{\psi_{t+1}(C_A)}{\bd_{t+1}(C_A)} \geq 0. 
\]
This implies that
\[
\psi_t(A) \geq 2 \psi_{t+1}(C_A),
\]
and hence that $\psi_{t+1}(C_A) \leq \frac{1}{2}\psi_t(A)$. To show a decrease in potential for $A \setminus C_A$, we will apply~\Cref{lem: progress lem} in a less trivial way. Let $\overline{\mu}_t = \inner{\mu_t^A}{\br_t}$. For $u \in A^\circ$, define $p_t(m_u) := \inner{\sum_{w \in A^\circ} \frac{\bM_t(u,w)}{\bM_t(u)} \frac{\bF_t(w)}{\bd(w)}}{\br_t}$. That is, $p_t(m_u)$ is the inner product of the linear combination of the flow vectors matched with $\bF_t(u)$ and $\br_t$. We need the following standard lemma about the expectation and concentration of these random inner products. 
\begin{lemma}\label{lem: concen lemma}
    For all $t$, we have
    \[
    \E[(p_t(u) - \overline{\mu}_t)^2] = \frac{1}{n} \norm{\frac{\bF_t(u)}{\bd(u)} - \mu_t^A}_2^2,
    \]
     \[
    \E[(p_t(u) - p_t(m_u))^2] = \frac{1}{n} \norm{\frac{\bF_t(u)}{\bd(u)} - \sum_{w \in A^\circ} \frac{\bM_t(u,w)}{\bM_t(u)} \frac{\bF_t(w)}{\bd(w)}}_2^2.
    \]
Moreover, each is at most $C \log n$ times its expectation (for constant $C > 0$) with high probability in $n$.
\end{lemma}
Note that by linearity of inner products and the definition of the cut and matching steps, we have the following.
\begin{claim} \label{cla: matched other side eta}
For all $t \geq 1$ and $u \in A^\circ$, we have
\[
(p_t(u) - p_t(m_u))^2 \geq (p_t(u) - \eta)^2.
\]
\end{claim}
\begin{proof}
Suppose that $p_t(u) \geq \eta$. We show that $p_t(m_u) \leq \eta$. By bilinearity of inner products, 
\begin{align*}
p_t(m_u) &= \sum_{w \in A^\circ} \frac{\bM_t(u, w)}{\bM_t(u)} \inner{\frac{{\bF}_{t-1}(w)}{\bd(w)}}{\br_t}  \\
&=  \sum_{w \in A^\circ} \frac{\bM_t(u, w)}{\bM_t(u)} p_t(w) \\
&\leq \sum_{w \in A^\circ}  \frac{\bM_t(u, w)}{\bM_t(u)} \eta \tag{By definition of the cut step} \\
&= \frac{\bM_t(u)}{\bM_t(u)} \eta
= \eta.
\end{align*}
The proof for the case of $p_t(u) < \eta$ is analogous.
\end{proof}
We are now ready to show that the potential decreases in $A \setminus C_A$ as well. 
\begin{lemma} \label{lem: prog A - C_A}
For each $A \in \calA_t' \subseteq \calA_t^\circ$, we have 
\[
\E[\psi_{t+1}(A \setminus C_A)] \leq (1 - \Omega(1/\log n)) \cdot \E[\psi_t(A)] + O(1/\poly n).
\]
\end{lemma}
\begin{proof}
With high probability, we have 
{
\allowdisplaybreaks
\begin{align*}
&\psi_t(A) - \psi_{t+1}(A \setminus C_A) \geq \frac{\psi_t(A)}{\bd_t(A)} - \frac{\psi_{t+1}(A \setminus C_A)}{\bd_{t+1}(A \setminus C_A)}\\
&\geq \frac{1}{2} \sum_{u \in (A \setminus C_A)^\circ} \bM_t(u)\left(1 - \frac{\bM_t(u)}{2 \bd(u)} \right) \norm{\frac{\bF_t(u)}{\bd(u)} - \sum_{w \in A^\circ} \frac{\bM_t(u,w)}{\bM_t(u)} \frac{\bF_t(w)}{\bd(w)}}_2^2 \\
 &+ \frac{1}{2} \sum_{u \in A^\circ \setminus (A \setminus C_A)^\circ} \bd(u) \norm{\frac{\bF_t(u)}{\bd(u)} - \mu_t^A}_2^2 \\
 &\geq \frac{1}{8} \sum_{u \in (A \setminus C_A)^\circ} \bd(u) \norm{\frac{\bF_t(u)}{\bd(u)} - \sum_{w \in A^\circ} \frac{\bM_t(u,w)}{\bM_t(u)} \frac{\bF_t(w)}{\bd(w)}}_2^2 + \frac{1}{2} \sum_{u \in A^\circ \setminus (A \setminus C_A)^\circ} \bd(u) \norm{\frac{\bF_t(u)}{\bd(u)} - \mu_t^A}_2^2\\
 &\geq \frac{|A^\circ|}{8C \log n} \sum_{u \in (A \setminus C_A)^\circ} \bd(u) (p_t(u) - p_t(m_u))^2 + \frac{|A^\circ|}{2C \log n} \sum_{u \in A^\circ \setminus (A \setminus C_A)^\circ} \bd(u) (p_t(u) - \overline{\mu}_t)^2 \\
  &\geq \frac{|A^\circ|}{8C \log n} \sum_{u \in S \cap (A \setminus C_A)^\circ} \bd(u) (p_t(u) - \eta)^2 + \frac{|A^\circ|}{2C \log n} \sum_{u \in S \cap (A^\circ \setminus (A \setminus C_A)^\circ)} \bd(u) (p_t(u) - \overline{\mu}_t)^2 \\
  &\geq \frac{|A^\circ|}{72C \log n} \sum_{u \in S} \bd(u) (p_t(u) - \overline{\mu}_t)^2 \\
  &\geq \frac{|A^\circ|}{2592C \log n} \sum_{u \in A^\circ} \bd(u) (p_t(u) - \overline{\mu}_t)^2.
\end{align*}
}
The first inequality uses $\bd_{t+1}(A \setminus C_A) \leq \bd_t(A)$. In the second, we apply~\Cref{lem: progress lem} with $S = A \setminus C_A$. Next, we use that $\bM_t(u) \leq \bd(u)$ and that we delete any source nodes in which $\bM_t(u) < \bd(u)/2$. The fourth inequality follows from an application of~\Cref{lem: concen lemma}. In the fifth inequality, we apply~\Cref{cla: matched other side eta} and restrict both sums to $u \in S \subseteq A^\circ$, where $S$ is the subset certifying progress from~\Cref{lem: progress set}. In the next inequality, we apply the first guarantee on $S$ from~\Cref{lem: progress set} and combine the sums. Finally, we apply the second guarantee on $S$ from~\Cref{lem: progress set}.
Note that by~\Cref{lem: concen lemma},
\begin{align*}
 \frac{|A^\circ|}{2592C \log n} \sum_{u \in A^\circ} \bd(u) \E[(p_t(u) - \overline{\mu}_t)^2] =  \frac{1}{2592C \log n} \sum_{u \in A^\circ} \bd(u) \norm{\frac{\bF_t(u)}{\bd(u)}- \mu_t^A}_2^2. 
\end{align*}
Hence, incorporating the high probability events above, we get
\[
\E[\psi_t(A) - \psi_{t+1}(A \setminus C_A)] \geq \Omega(\psi_t(A)/\log n) - O(1/\poly n),
\]
as needed.
\end{proof}
The combination of the immediate consequence of~\Cref{lem: progress lem} for $C_A$ and~\Cref{lem: prog A - C_A} is the following.
\begin{corollary} \label{cor: counter threshold}
Let $A \in \calA_t$ such that $\bd_t(A) > 0$, and recall the variable $x^A_t$ from \Cref{alg:weak-ed-del}. If $x^A_t > 10^5 C \log n \log nW$, then $\psi_t(A) < 1/(mW)^8$ with high probability in $n$.
\end{corollary}

To complete the proof of the main result of this section, we also need to show that $\bd(V) - \bd_T(V) \leq \deleps \bd(V)$. 

\begin{lemma} \label{lem: del bd}
For all $t \geq 0$, we have 
\[
\bd(V) - \bd_t(V) \leq 64t \floweps \bd(V).
\]
\end{lemma}
\begin{proof}
Recall that we either have $\bd_t(u) = \bd(u)$ or $\bd_t(u) = 0$ for each $t \geq 0$ and $u \in V$. We just need to bound the amount of total demand of deleted vertices in the latter case. In~\Cref{alg:weak-ed-del}, we delete vertices in only two cases. 

First, if there exists $s \leq t$ such that $u \in L_A \cap A \setminus C_A \subseteq A \in \calA_s'$ and we route less than half of its source in the corresponding matching step, we delete it. Secondly, if the total proportion of non-deleted demand in some component in the partition is at most $15/16$ of the total demand, we delete all of the demand in the component. We can charge deletions of the second kind to the first kind, incurring a factor of $16$. 

So, it remains to bound the amount of demand deleted from unrouted source. For this, we can appeal to the second property of the matching player flow oracle: we route at least $\Delta(A \setminus C_A) - 2\floweps \bd(A)$ of the source in $L_A \cap A \setminus C_A$. As such, the amount of deleted demand can only increase by $4\floweps \bd(V) \leq 4\floweps \bd(V)$ from deletions of the first kind on each step. Summing over all the steps yields the desired bound.
\end{proof}
We can now put all of this together to prove our first main result.
\begin{theorem}[Weak Expander Decomposition with Partial Deletions] \label{thm: weak-decomp-del}
Given $G=(V, E, \bc), \bd \in \R_{\geq 0}^V, \phi > 0$, and access to an approximate max flow oracle as in~\Cref{ora: matching player} with parameter $1 \geq \floweps > 0$, running in time $R(n, m, \floweps)$ per query, there is an algorithm running $T = O(\log n \log nW)$ rounds of cut-matching which computes a partition $\calA_T$ of $V$ and $\bd_T \leq \bd$ with the following properties:
\begin{enumerate}
    \item For each $A \in \calA_T$ with $\bd_T(A) > 0$, $\bd_T(A) \geq 15\bd(A)/16$. Moreover, with high probability $\{\bd_T|_A \,:\, A \in \calA\}$ mix simultaneously in $G$ with congestion $4T/\phi$.
    \item The total capacity of edges cut by $\calA_T$ is at most $O((\phi \log nW + \floweps T)\bd(V))$.
       \item $\bd(V) - \bd_T(V) \leq 64T \floweps \bd(V)$.
    \item The algorithm runs in time $O(T(R(n,m, \floweps) + mT))$.
\end{enumerate}
\end{theorem}
\begin{proof}
We output $\calA_T$ from~\Cref{alg:weak-ed-del}.
To prove (1), consider the following potential function 
\[
\Phi_t = \sum_{A \in \calA_t^\circ} (10^5 C \log n \log nW - x_t^A) \bd_t(A) \geq 0.
\]
In each iteration,~\Cref{ora: matching player} outputs $\calA'_t \subseteq \calA^\circ_t$ with 
\[
 \bd_{t}\left(\bigcup_{A \in \calA'_{t}} A\right) \geq \frac{1}{2} \bd_{t}\left(\bigcup_{A \in \calA_{t}^\circ} A\right). 
\]
Consequently, $\Phi_t$ decreases by a (multiplicative) factor of at least $1- \frac{1}{2 \cdot 10^5 C \log n \log nW}$ in each iteration. Hence, after $T$ iterations (for $T = O(\log n \log nW)$, with high probability we have $\Phi_t = 0$. At that point, $\calA_T = \calA_T^\times.$ Thus, for all $A \in \calA_T^\times$, with $A^\circ \neq \emptyset$, by~\Cref{cor: counter threshold}, we have $\psi_T(A) < 1/(mW)^8$ or $|\supp(\bd_T|_{A})| = 1$. Note that, by~\Cref{alg:weak-ed-del}, each such $A$ with $A^\circ \neq \emptyset$ must have $\bd_T(A) \geq 15\bd(A)/16.$ Moreover, by~\Cref{lem: progress set}, we have that $\{\bd_T|_A \,:\, A \in \calA_T\}$ mix simultaneously in $G$ with congestion $4T/\phi$.  This yields (1).

To prove (2), note first that, by the first property of~\Cref{ora: matching player}, the total capacity of the computed cuts in a given round is at most
\[
\frac{\phi}{8} \sum_{A \in \calA'_{t-1}} \bd_{t-1}(C_A) + 2\floweps \bd(V) \leq \frac{\phi}{8} \sum_{A \in \calA'_{t-1}} \bd(C_A) + 2\floweps \bd(V).
\]
The latter term contributes at most $2\floweps T \bd(V)$ over all the rounds. To account for the first term, in each iteration we charge each vertex $v \in C_A$ a total of $\frac{\phi}{8} \bd(v)$. Since we have the additional guarantee in first property of~\Cref{ora: matching player} that $\bd_{t-1}(C_A) \leq \bd_{t-1}(A)/2$, each vertex can only be charged $O(\log nW)$ times this way. Each intercluster edge in $\calA_T$ is accounted for in this argument, yielding (2).

(3) is immediate from~\Cref{lem: del bd}. 

Finally, for (4), the running time of the algorithm comes from $T$ cut steps and $T$ matching steps. The $t$\textsuperscript{th} cut step takes $O(mt) \leq O(mT)$ time from~\Cref{lem: cut step runtime}. Each matching step takes $O(R(n,m,\floweps))$ time from the flow oracle call, plus an additional $O(m\log m)$ time for using link-cut trees to find a flow path decomposition of the matching flow (to form $\bM_t$). Since $T \geq \log m$, the link-cut tree runtime is subsumed by the cut step running time. Summing these bounds over all $T$ rounds of the algorithm yields the desired result.
\end{proof}
\subsection{Grafting in Deleted Demand} \label{sec: grafting}
One potential weakness of the partition from~\Cref{thm: weak-decomp-del} is property (1), its mixing guarantee. It is not quite the case that every $A \in \calA_T$ is either certified as a (simultaneously mixing) $(\phi, \bd)$-near-expander or entirely deleted. Instead, the expanding components might have some deleted nodes inside them still (i.e., $\bd_T(A) < \bd(A)$). Moreover, it might be the case that we want some stronger notion of expansion, e.g., boundary-linked expansion~\cite{goranci2021expander}. 

Fortunately, we can strengthen the decomposition with one additional \textit{grafting} step, similarly to~\cite{directed-expander-decomposition}. Let $\deg_{\partial \calA_T} : V \to \Z_{\geq 0}$ be the additional vertex weighting on $G$ corresponding to the boundary of $\calA_T$. Let $\psi > 0$ be a parameter. Think of $\psi$ as $\phi$ in the case of expander decompositions; for our application, we will set $\psi = \Omega(1)$.  Consider the flow instance on a subgraph of $G$, generated as follows.
\begin{itemize}
    \item Let $\calA_T^+ = \{A \in \calA_T \,:\, \bd_T(A) > 0\}$.
    \item For each $A \in \calA_T^+$:
    \begin{itemize}
        \item For each $u \in A$, add $\Delta(u) = \deg_{\partial \calA_T}(u) + \bd(u) - \bd_T(u)$ source. 
        \item For each $u \in A$ with $\bd_T(u) = \bd(u)$, add sink $\nabla(u) = \bd(u)/5$. 
    \end{itemize}
    \item Remove all edges cut by $\calA_T$ from $G$, and scale the capacity of remaining edges by $1/\psi$. 
\end{itemize}

The intuition for this flow instance is the following. If it is feasible, then we can route all the deleted and boundary demand to non-deleted demand that is certified to mix. (We set the sinks as $\bd(u)/5$ rather than $\bd(u)$ purely to streamline our specific application.) As such, at the cost of a slight increment in the congestion, we certify that all the expanding components $A$ mix simultaneously with respect  to demands $(\bd + \deg_{\partial \calA_T}(u))|_A$, not just $\bd_T|_A$. If the flow is not feasible, because we scaled the edges and since the source is small relative to the sink, we will find sparse cuts in most components. We want the additional guarantee that any remaining source is almost entirely routed and the new boundary can be routed as well. This condition is achievable with fair cut-based max flow algorithms (e.g.,~\cite{DBLP:conf/soda/0006NPS23,DBLP:conf/ipco/LiL25}), and we require something analogous in the definition of our flow oracle.

\begin{oracle}[Grafting Flow Oracle] \label{ora: grafting flow}
On such a flow instance, for some parameter $\floweps > 0$, we find a flow with the following properties: 
\begin{enumerate}
\item For each $A \in \calA_T^+$, if $\deg_{\partial \calA_T}(A) \leq \bd_T(A)/8$, we find a pair $(C_A,A \setminus C_A)$ such that:
\begin{enumerate}
    \item For each $u \in A \setminus C_A$ with $\Delta(u) > 0$, we route at least $(1 - \floweps)\Delta(u)$ source from $u$ to $A \setminus C_A$.
    \item The flow saturates at least a $(1 - \floweps)$ fraction of the capacity of each edge from $C_A$ to $A \setminus C_A$.
    \item We have $\sum_{A} \bc_G(E(C_A, A \setminus C_A)) \leq 8\psi \bd(V).$
\end{enumerate}   
\item We have 
\[
\bd \left( \bigcup_{A \in \calA_T^+} C_A \right) \leq 30 (\bd(V) - \bd_T(V) + \deg_{\partial \calA_T}(V)).
\]
\end{enumerate}
\end{oracle}

We can now state the main result of this section, a strengthening of~\Cref{thm: weak-decomp-del} using a grafting post-processing step.

\begin{restatable}[Weak Expander Decomposition with Deletions]{theorem}{maindecomp} \label{thm: weak-decomp-del2}
Suppose we have $G=(V, E, \bc)$, $\bd \in \Z_{\geq 0}^V$, $\phi > 0$, $\psi > 0$, and access to~\Cref{ora: matching player}  with parameter $1 \geq \floweps_1 > 0$ and~\Cref{ora: grafting flow} with parameter $\floweps_2 \leq 1/10$, running in time $R_1(n, m, \floweps_1)$ and $R_2(n, m, \floweps_2)$ per query, respectively. Let $T = O(\log n \log nW)$. Then, there is an algorithm computing a partition $\calA = \calA^\circ \sqcup \calA^\times$ of $V$ with the following properties:
\begin{enumerate}
 \item The algorithm runs in time $O(T(R_1(n,m, \floweps_1) + mT) + R_2(n,m, \floweps_2))$.
      \item $\bd(\bigcup_{A \in \calA^\times} A) = O((\floweps_1 T + \phi \log nW)\bd(V)).$
        \item The total capacity of edges cut by $\calA$ is at most $O((\phi \log nW  + \floweps_1 T + \psi)\bd(V))$.
    \item $\{(\bd + \deg_{\partial \calA})|_A \,:\, A \in \calA^\circ\}$ mix simultaneously in $G$ with congestion $T/\phi + \frac{2}{\psi}$.
        \item There exists a flow of congestion $\frac{2}{\psi}$ such that each $u \in A \in \calA^\circ$ sends $\deg_{\partial \calA}(u)$ flow and each $v \in V$ receives at most $\bd(v)/4$ flow.
\end{enumerate}
\end{restatable}
\begin{proof}
 First we compute a partition $\calA_T$ of $G$ with respect to $\phi$ and $\bd$, using~\Cref{thm: weak-decomp-del} and~\Cref{alg:weak-ed-del}. Then, we construct the flow instance above with respect to $\psi$ and call~\Cref{ora: grafting flow}. For each $A \in \calA_T^+$ with $\deg_{\partial \calA_T}(A) \leq \bd_T(A)/8$,~\Cref{ora: grafting flow} outputs a pair $(C_A, A \setminus C_A)$. If $C_A$ is nonempty, we add $C_A$ to $\calA^\times$. We add $A \setminus C_A$ to $\calA^\circ$. For each $A \in \calA_T^+$ with $\deg_{\partial \calA_T}(A) > \bd_T(A)/8$ and $A \in \calA_T \setminus \calA_T^+$, we add $A$ to $\calA^\times$. All together this forms $\calA$. By~\Cref{thm: weak-decomp-del}, this algorithm runs in time $O(T(R_1(n,m, \floweps_1) + mT) + R_2(n,m, \floweps_2))$, yielding (1).

By the second property of~\Cref{ora: grafting flow} and the third property of~\Cref{thm: weak-decomp-del}, we get that 
\begin{align*}
\bd\left( \bigcup_{A \in \calA_T^\times } A\right) + \bd \left( \bigcup_{A \in \calA_T^+} C_A \right) &\leq (\bd(V) - \bd_T(V)) +  30 (\bd(V) - \bd_T(V))  + 30\deg_{\partial \calA_T}(V)\\
&\leq  2000T\floweps_1\bd(V) + 30\deg_{\partial \calA_T}(V).
\end{align*}
For (2), we also need to account for the sum of $\bd(A)$ over $A \in \calA_T^+$ with 
\[
\deg_{\partial \calA_T}(A) > \bd_T(A)/8 > 15\bd(A)/144,
\]
where the last inequality uses that $A \in \calA_T^+$, so $\bd_T(A) > 15 \bd(A)/16$. Hence, summing over all such $A$, we get 
\[
\sum_{\substack{A \in \calA_T^+ \\
\deg_{\partial \calA_T}(A) > \bd_T(A)/8}} \bd(A)  < 10\deg_{\partial \calA_T}(V).
\]
 Combining these bounds and using the second property of~\Cref{thm: weak-decomp-del} yields (2). 

All the cut edges between $A \in \calA_T$ are already accounted for by the second property of~\Cref{thm: weak-decomp-del}. So, we just need to bound the capacity of the cuts $(C_A, A)$. But by~\Cref{ora: grafting flow}, this is at most $8\psi \bd(V)$. This gives (3). 

Finally, we prove (4). Let $\calA^\circ = \{A_1', A_2', \ldots A_r'\}$. For each $i \in [r]$, let $\bb_i$ be a demand respecting $(\bd + \deg_{\partial \calA})|_{A_i'}$. Each $A_i'$ belongs to some set in $\calA_T$; denote that set by $A_i$. Decompose $\bb_i = \bb_i^{(1)} + \bb_i^{(2)}$ where $\bb_i^{(1)}$ respects $(\bd + \deg_{\partial \calA_T})|_{A_i'}$ and $\bb_i^{(2)}$ respects $(\deg_{\partial \calA} - \deg_{\partial \calA_T}) |_{A_i'}$. The idea is to route each $\bb_i$ to a residual demand respecting $\bd_T |_{A_i}$; then we can appeal to the first property of~\Cref{thm: weak-decomp-del} to finish the routing.

Let $f$ be the flow computed by~\Cref{ora: grafting flow}. Take a path decomposition of $f$. So, $f(u,v)$ denotes the amount of flow sent from $u$ to $v$ via paths in the path decomposition of $f$. Let $f(u) := \sum_{v \in V} f(u,v)$. We can ensure that $f(u, v)$ is nonzero only if $\Delta(u) - \nabla(u) > 0$ and $\nabla(v) - \Delta(v) > 0$. Since~\Cref{ora: grafting flow} guarantees that edges crossing the cut $(C_A, A \setminus C_A)$ are at least $(1 - \floweps_2)$ fraction saturated, the net flow across them is at least $(1 - 2\floweps_2)$ their capacity. Since $\floweps_2 \leq 1/5 < 1/2$, we can in addition assume that each path in the path decomposition has at most one edge crossing the cut.

We can use $f$ to route $\bb_i$ for each $i$. First, for each $u \in A_i'$ and $v \in V$, route $\frac{f(u,v) \bb_i^{(1)}(u)}{f(u)}$ flow, using flow paths in the path decomposition starting in $A_i'$ and remaining in $A_i'$. Additionally, for each $u \in A_i'$, route $\bb_i^{(2)}(u)$ total flow using flow paths from the path decomposition which cross $(C_{A_i}, A_i \setminus C_{A_i})$ via $u$, scaling flow paths by at most $1/(1-2\floweps_2)$.
We can route all of these flows simultaneously with congestion $\frac{1}{(1-2\floweps_2)\psi}$ since $f$ is routable in $G$ with congestion at most $1/\psi$. This uses that we route at least a $(1-\floweps_2)$ fraction of source for each $u \in A \setminus C_A$ with $\Delta(u) > 0$ and the edges crossing the cuts are at least a $(1 - \floweps_2)$ fraction saturated.

After this routing, for each $u \in A_i'$ with $\Delta(u) - \nabla(u) > 0$, we routed $\bb_i(u)$ total flow. Each $v \in A_i$ with $\bd_T(v) = \bd(v)$ receives at most
\[
\frac{\nabla(v)}{1-2\floweps_2} = \frac{\bd(v)}{5(1-2\floweps_2)} \leq \frac{\bd(v)}{4}
\]
total demand, since we scaled flow paths by at most $1/(1-2\floweps_2)$ and $\floweps_2 \leq 1/10$.

Furthermore, the residual demand $\bb_i - \bb_i'$ satisfies $(\bb_i - \bb_i')(A_i) = 0$, $\bb_i(u) - \bb_i'(u) = 0$ for $u \in A_i$ with $\bd_T(u) = 0$, and 
\[
|\bb_i - \bb_i'|  \leq \bd/4.
\]
That is, $\bb_i(u) - \bb_i'(u)$ is a demand respecting $\frac{1}{4}\bd|_{A_i}$ and hence $\{(\bb_i(u) - \bb_i'(u))|_{A_i} \in \calA_T\}$ can be routed with congestion $T/\phi$. In total then we can route
\[
\{\bb_i|_{A_i'} \,:\, i \in [r]\} =\{(\bb_i - \bb_i' + \bb_i')|_{A_i'} \,:\, i \in [r]\}
\]
with total congestion at most $T/\phi + \frac{2}{\psi}$, using that $\floweps_2 \leq 1/4$.

(5) follows directly from the same argument as constructing the routing of $\bb_i - \bb_i'$.

\end{proof}

\section{Sufficient Conditions for Constructing a Congestion-Approximator}\label{sec:proof-of-congestion-approximator}

This section is dedicated to proving that the following properties suffice to obtain a congestion-approximator.

\begin{theorem}\label{thm:congestion-approximator-guarantee}
    Consider a capacitated graph $G=(V,E, \bc)$, and let $\alpha\ge 1$ and $\beta\ge1$ be parameters. Consider a sequence of partitions $\mathcal{P}_1,\ldots,\mathcal{P}_L$ of $V_1,\ldots,V_L\subseteq V$. For ease of notation, let $\overline{\mathcal{P}}_0$ denote the singleton partition. For each $i\in[L]$, define $\mathcal{Q}_i$ to be the induced partition of $\overline{\mathcal{P}}_{i-1}$ on $V\setminus V_i$ i.e., 
    $$\mathcal{Q}_i=\{C\cap (V\setminus V_i):C\in\overline{\mathcal{P}}_{i-1},C\cap (V\setminus V_i)\neq\emptyset\},$$
    and let $\overline{\mathcal{P}}_i=\mathcal{P}_i\cup\mathcal{Q}_i$ for each $i\in[L]$. Suppose the partitions $\overline{\mathcal{P}}_1,\ldots,\overline{\mathcal{P}}_L$ satisfy:
        \begin{enumerate}
            \item $\overline{\mathcal{P}}_1$ is a partition $\{\{v\}:v\in V\}$ of singleton clusters, and $\overline{\mathcal{P}}_L$ is a partition $\{\{V\}\}$ with a single cluster.
            \item For each $i\in[L-1]$, the collection of vertex weightings $\{\deg_{\partial\overline{\mathcal{P}}_i\cup\partial C}|_C\in\mathbb{R}_{\ge0}^V:C\in\overline{\mathcal{P}}_{i+1},C\subseteq V_{i+1}\}$ mixes simultaneously in $G$ with congestion $\alpha$. 
            \item For each $i\in[L-1]$, there is a flow in $G$ with congestion $\beta$ such that each $v\in V_{i+1}$ sends $\deg_{\partial\overline{\mathcal{P}}_{i+1}}(v)$ flow and receives at most $\frac{1}{4}\deg_{\partial\overline{\mathcal{P}}_i}(v)$ flow.
        \end{enumerate}
	For each $i\in[L]$, let partition $\mathcal{R}_{\ge i}$ be the common refinement of partitions $\overline{\mathcal{P}}_i,\ldots,\overline{\mathcal{P}}_L$, i.e.,
	\begin{align*}
		\mathcal{R}_{\ge i}=\{C_i\cap\cdots\cap C_L:C_i\in\overline{\mathcal{P}}_i,\ldots,C_L\in\overline{\mathcal{P}}_L,C_i\cap\cdots\cap C_L\neq\emptyset\}.
	\end{align*}
	Then, their union $\mathcal{C}=\bigcup_{i\in[L]}\mathcal{R}_{\ge i}$ is a congestion-approximator with quality $48\alpha\beta L^2$.
\end{theorem}

Note that most in most natural algorithms, we have $V_L=V$ in the top layer. The rest of the section proves the theorem. As in~\cite{DBLP:conf/soda/0006R025}, we will actually need a \textit{pseudo-congestion-approximator} analogue of~\Cref{thm:congestion-approximator-guarantee}, where $\overline{\mathcal{P}}_L$ is not necessarily the partition $\{\{V\}\}$. The precise guarantees are given below. In particular, note that assumptions (2) and (3) remain unchanged.

\begin{lemma}\label{lem:congestion-approximator-guarantee}
    Consider a capacitated graph $G=(V,E, \bc)$ with $\bc \in [1, W]\cap \Z$, and let $\alpha\ge 1$ and $\beta\ge1$ be parameters. Consider a sequence of partitions $\mathcal{P}_1,\ldots,\mathcal{P}_L$ of $V_1,\ldots,V_L\subseteq V$. For ease of notation, let $\overline{\mathcal{P}}_0$ denote the singleton partition. For each $i\in[L]$, define $\mathcal{Q}_i$ to be the induced partition of $\overline{\mathcal{P}}_{i-1}$ on $V\setminus V_i$ i.e., 
    $$Q_i=\{C\cap (V\setminus V_i):C\in\overline{\mathcal{P}}_{i-1},C\cap (V\setminus V_i)\neq\emptyset\},$$
    and let $\overline{\mathcal{P}}_i=\mathcal{P}_i\cup\mathcal{Q}_i$ for each $i\in[L]$. Suppose the partitions $\overline{\mathcal{P}}_1,\ldots,\overline{\mathcal{P}}_L$ satisfy:
        \begin{enumerate}
            \item $\overline{\mathcal{P}}_1$ is the partition $\{\{v\}:v\in V\}$ of singleton clusters.
            \item For each $i\in[L-1]$, the collection of vertex weightings $\{\deg_{\partial\overline{\mathcal{P}}_i\cup\partial C}|_C\in\mathbb{R}_{\ge0}^V:C\in\overline{\mathcal{P}}_{i+1},C\subseteq V_{i+1}\}$ mixes simultaneously in $G$ with congestion $\alpha$.
            \item For each $i\in[L-1]$, there is a flow in $G$ with congestion $\beta$ such that each $v\in V_{i+1}$ sends $\deg_{\partial\overline{\mathcal{P}}_{i+1}}(v)$ flow and each $v\in V_{i+1}$ receives at most $\frac{1}{4}\deg_{\partial\overline{\mathcal{P}}_i}(v)$ flow.
        \end{enumerate}
    For each $i\in[L]$, let partition $\mathcal{R}_{\ge i}$ be the common refinement of partitions $\overline{\mathcal{P}}_i,\ldots,\overline{\mathcal{P}}_L$, i.e.,
	\begin{align*}
		\mathcal{R}_{\ge i}=\{C_i\cap\cdots\cap C_L:C_i\in\overline{\mathcal{P}}_i,\ldots,C_L\in\overline{\mathcal{P}}_L,C_i\cap\cdots\cap C_L\neq\emptyset\}.
	\end{align*}
    Consider their union $\mathcal{C}=\bigcup_{i\in[L]}\mathcal{R}_{\ge i}$. For any demand $\mathbf{b}\in\mathbb{R}^V$ satisfying $|\mathbf{b}(C)|\le\delta C$ for all $C\in\mathcal{C}$, there exists a demand $\mathbf{b}'\in\mathbb{R}^V$ satisfying $|\mathbf{b}'|\le\deg_{\partial\overline{\mathcal{P}}_L}$ and a flow routing $\mathbf{b}-\mathbf{b}'$ with congestion $48\alpha\beta L^2$.
\end{lemma}

Instead of proving \Cref{thm:congestion-approximator-guarantee} directly, we will prove \Cref{lem:congestion-approximator-guarantee}, which is needed for the algorithm. Below, we give the proof \Cref{thm:congestion-approximator-guarantee} assuming \Cref{lem:congestion-approximator-guarantee}.
\begin{proof}[Proof of \Cref{thm:congestion-approximator-guarantee}]
    Consider the partitions $\overline{\mathcal{P}}_1,\ldots,\overline{\mathcal{P}}_L$ that satisfy the assumptions of \Cref{thm:congestion-approximator-guarantee}. For a given demand vector $\mathbf{b}\in\mathbb{R}^V$ satisfying $|\mathbf{b}(C)|\le\delta C$ for all $C\in\mathcal{C}$, we want to establish a flow routing demand $\mathbf{b}$ with congestion $48\alpha\beta L^2$.~\Cref{thm:congestion-approximator-guarantee} then follows by the definition of a congestion-approximator. 

    Apply \Cref{lem:congestion-approximator-guarantee} to the partitions $\overline{\mathcal{P}}_1,\ldots,\overline{\mathcal{P}}_L$ and demand $\mathbf{b}$. We obtain a demand $\mathbf{b}'\in\mathbb{R}^V$ satisfying $|\mathbf{b}'|\le\deg_{\partial\overline{\mathcal{P}}_L}$ and a flow $f$ routing demand $\mathbf{b}-\mathbf{b}'$ with congestion $48\alpha\beta L^2$. By assumption (1) of \Cref{thm:congestion-approximator-guarantee}, we have $\overline{\mathcal{P}}_L=\{\{V\}\}$ which implies $\partial\overline{\mathcal{P}}_L=\emptyset$. Since $|\mathbf{b}'|\le\deg_{\partial\overline{\mathcal{P}}_L}=\mathbf{0}$, we must have $\mathbf{b}'=\mathbf{0}$. It follows that flow $f$ routes demand $\mathbf{b}$ with congestion $48\alpha\beta L^2$, completing the proof.
\end{proof}

In the rest of the section, we prove \Cref{lem:congestion-approximator-guarantee}. We first begin with a few observations about the structure of the partitions $\mathcal{R}_{\ge i}$, which will be needed later. The first two are direct analogues of Claim 4.3 and 4.4 of~\cite{DBLP:conf/soda/0006R025}, and their proofs are identical.
\begin{observation}\label{obs:refinement}
    For all $i,j\in[L]$ with $i<j$, the partition $\mathcal{R}_{\ge i}$ of $V$ is a refinement of the partition $\mathcal{R}_{\ge j}$. That is, each set in $\mathcal{R}_{\ge j}$ is a disjoint union of sets in $\mathcal{R}_{\ge i}$. In particular, $\partial\mathcal{R}_{\ge i}\supseteq \partial\mathcal{R}_{\ge j}$.
\end{observation}
\begin{proof}
    Consider a set $C=C_j\cap\cdots\cap C_L\in\mathcal{R}_{\ge j}$ for some $C_j\in\overline{\mathcal{P}}_j,\ldots,C_L\in\overline{\mathcal{P}}_L$. Since $\overline{\mathcal{P}}_i,\ldots,\overline{\mathcal{P}}_{j-1}$ are all partitions of $V$, the set $C$ is the disjoint union of all non-empty sets of the form $C_i\cap\ldots\cap C_{j-1}\cap C$ for $C_i\in\overline{\mathcal{P}}_i,\ldots,C_{j-1}\in\overline{\mathcal{P}}_{j-1}$. Therefore, $\mathcal{R}_{\ge i}$ is a refinement of $\mathcal{R}_{\ge j}$. Since refinements can only increase the boundary set, the second statement $\partial\mathcal{R}_{\ge i}\supseteq\partial\mathcal{R}_{\ge j}$ follows.
\end{proof}

\begin{observation}\label{obs:difference-of-R}
    For all $i\in[L-1]$, we have $\partial\mathcal{R}_{\ge i}\setminus\partial\mathcal{R}_{\ge i+1}\subseteq \partial\overline{\mathcal{P}}_{i}$.
\end{observation}
\begin{proof}
    Consider an edge $(u,v)\in\partial\mathcal{R}_{\ge i}\setminus\partial\mathcal{R}_{\ge i+1}$. Since $(u,v)\not\in\partial\mathcal{R}_{\ge i+1}$, there must exist a set $C\in\mathcal{R}_{\ge i+1}$ which contains both $u$ and $v$. Since $\overline{\mathcal{P}}_i$ is a partition of $V$, the set $C$ is the disjoint union of non-empty sets of the form $C\cap C_i$, for $C_i\in\overline{\mathcal{P}}_i$. Since $u,v\in C$, we know $u$ and $v$ each belong to some set of the form $C\cap C_i$. But since $u,v\in\partial\mathcal{R}_{\ge i}$, they must lie in different sets. These sets can only differ in the choice of $C_i\in\overline{\mathcal{P}}_i$, so we have that $u$ and $v$ lie in different sets in partition $\overline{\mathcal{P}}_i$, implying $(u,v)\in\partial\overline{\mathcal{P}}_i$, as desired. 
\end{proof}

\begin{observation}\label{obs:equal-degrees-not-in-V-i+1}
    For $i\in{0}\cup[L-1]$, we have that if $u\not\in V_{i+1}$, then:
    \begin{enumerate}
    \item $\deg_{\partial\overline{\mathcal{P}}_{i+1}\setminus \partial V_{i+1}}(u)\le \deg_{\partial\overline{\mathcal{P}}_{i}}(u)$
    \item $\deg_{\partial\mathcal{R}_{\ge i}}(u)=\deg_{\partial\mathcal{R}_{\ge i+1}}(u)$.
    \end{enumerate}
\end{observation}
\begin{proof}
First, we prove property (1). If $u \not\in V_{i+1}$, then note that $u$ belongs to a set in $\mathcal{Q}_{i+1}$ in $\overline{\mathcal{P}}_{i + 1}$. Hence, by definition of $\mathcal{Q}_{i+1}$, there exists $C \in \overline{\mathcal{P}}_{i}$ such that $u \in C \cap (V \setminus V_{i+1})$ and $C \cap (V \setminus V_{i+1}) \in \overline{\mathcal{P}}_{i + 1}$. Consider that any edge $(u,v)$ that is part of the boundary $\partial\overline{\mathcal{P}}_{i+1}\setminus \partial V_{i+1}$. This means that $(u,v)$ crosses the boundary of $C\cap (V\setminus V_{i+1})\in\overline{\mathcal{P}}_{i+1}$ but does not cross through the boundary of $V_{i+1}$, which implies that $(u,v)$ must cross through the boundary of $C$. But this implies that $(u,v)$ also crosses through the boundary $\partial\overline{\mathcal{P}}_i$ since $C\in\overline{\mathcal{P}}_i$, completing the proof.

Next, consider property (2). Again, we have that there exists $C \in \overline{\mathcal{P}}_{i}$ such that $u \in C \cap (V \setminus V_{i+1})$ and $C \cap (V \setminus V_{i+1}) \in \overline{\mathcal{P}}_{i + 1}$. But we also have that $u$ belongs to $C$ in $\overline{\mathcal{P}}_{i}$. Then, since $C \cap (V \setminus V_{i+1}) \subseteq C$, for any $(u,v) \in \partial \overline{\mathcal{P}}_{i}$, we must have $(u,v) \in \partial \overline{\mathcal{P}}_{i + 1} \subseteq \partial \mathcal{R}_{\ge i+1}$. In other words, any edge in $\partial \overline{\mathcal{P}}_{i}$ that contains $u$ is also in $\partial \mathcal{R}_{\ge i+1}$. 
Since $\partial\mathcal{R}_{\ge i}\setminus\partial\mathcal{R}_{\ge i+1}\subseteq \partial\overline{\mathcal{P}}_{i}$ by~\Cref{obs:difference-of-R}, we conclude that there are no edges in $\partial\mathcal{R}_{\ge i}\setminus\partial\mathcal{R}_{\ge i+1}$ containing $u$, concluding the second part of the claim. 
\end{proof}

Now, let $\mathbf{b}\in\mathbb{R}^V$ be a demand satisfying $|\mathbf{b}(C)|\le\delta C$ for each $C\in\mathcal{C}$. Our goal is to construct a demand $\mathbf{b}'\in\mathbb{R}^V$ satisfying $|\mathbf{b}'|\le\deg_{\partial\mathcal{P}_L}$ and a flow routing demand $\mathbf{b}-\mathbf{b}'$ with congestion $48\alpha \beta L^2$. We follow the same high-level strategy as~\cite{DBLP:conf/soda/0006R025}. The flow is constructed over $L-1$ iterations. On iteration $i\in[L-1]$, we construct a flow $f_i$ and demand $\mathbf{b}_i$ such that:
\begin{itemize}
    \item[1.] $f_i$ routes $\mathbf{b}_{i-1}-\mathbf{b}_i$, where we define $\mathbf{b}_0=\mathbf{b}$ for iteration $i=1$,
    \item[2.] $f_i$ has congestion $48L\alpha\beta$,
    \item[3.] for each $C\in\mathcal{R}_{\ge i+1}$, we have $(\mathbf{b}_{i-1}-\mathbf{b}_i)(C)=0$, and
    \item[4.] $\mathbf{b}_i$ satisfies $|\mathbf{b}_i|\le\deg_{\partial\mathcal{R}_{\ge i+1}}$.
\end{itemize}

Properties (1), (2), and (4) alone are sufficient to prove~\Cref{lem:congestion-approximator-guarantee} with demand $\mathbf{b}'=\mathbf{b}_{L-1}$ and flow $f_1+\cdots+f_{L-1}$. Indeed, noting that $\partial\mathcal{R}_{\ge L}=\partial\overline{\mathcal{P}}_L$ by definition, we have the following.

\begin{observation}
    Suppose that properties (1), (2), and (4) hold for each $i\in[L-1]$. Then the demand $\mathbf{b}_{L-1}$ satisfies $|\mathbf{b}_{L-1}|\le\deg_{\partial\overline{\mathcal{P}}_L}$ and the flow $f_1+\cdots+f_{L-1}$ routes demand $\mathbf{b}-\mathbf{b}_{L-1}$ with congestion $48\alpha\beta L^2$.
\end{observation}

Despite the fact that (3) is not directly necessary to establish~\Cref{lem:congestion-approximator-guarantee}, we will use this property in our iterative construction. In order to establish the conditions (1)--(4) above, we will use the following technical lemma.

\begin{restatable}{lemma}{techlem}\label{lem:main-technical-lemma}
    Consider an iteration $i\in[L-1]$ and a vector $\mathbf{s}\in\mathbb{R}^V$ such that
    \begin{itemize}
        \item[(a)] $|\mathbf{s}|\le\deg_{\partial\mathcal{R}_{\ge i}}$ and
        \item[(b)] $|\mathbf{s}(C)|\le\delta C$ for all $C\in\mathcal{R}_{\ge i+1}$.
    \end{itemize}
    Then, we can construct a flow $f$ such that
    \begin{itemize}
        \item[(i)] $f$ routes demand $\mathbf{s}-\mathbf{t}$ for vector $\mathbf{t}\in\R^V$ with $|\mathbf{t}|\le\deg_{\partial\mathcal{R}_{\ge i+1}}$, 
        \item[(ii)] $f$ has congestion $48L\alpha\beta$, and
        \item[(iii)] for all $C\in\mathcal{R}_{\ge i+1}$, we have $(\mathbf{s}-\mathbf{t})(C)=0$.
    \end{itemize}
\end{restatable}

Before proving the lemma, we first establish that it implies properties (1)--(4) for suitable $f_i$ and $\mathbf{b}_i$. This proof is identical to the proof of Lemma 4.7 in~\cite{DBLP:conf/soda/0006R025}.

\begin{lemma}
    Assuming \Cref{lem:main-technical-lemma}, we can construct $f_i$ and $\mathbf{b}_i$ satisfying properties (1)--(4) for $i\in[L-1]$.
\end{lemma}

\begin{proof}
    We induct on $i\in[L-1]$, with the base case being $i=0$. In this base case, we set $f_0$ as the empty flow, which routes $\mathbf{b}-\mathbf{b}_0=0$, so properties (1)--(3) follow trivially. For property (4), observe that the singleton sets $\{v\}$ are in $\overline{\mathcal{P}}_1$, so they are also in $\mathcal{C}$. This implies that $|\mathbf{b}(\{v\})|\le\deg(v)$ for all $v\in V$, which thereby implies that $|\mathbf{b}_0|=|\mathbf{b}|\le \deg=\deg_{\partial\mathcal{R}_{\ge 1}}$, as desired.

    For the inductive step, we apply \Cref{lem:main-technical-lemma} on $i\ge 1$ and $\mathbf{s}=\mathbf{b}_{i-1}$. We first verify the conditions on $\mathbf{s}$ required by \Cref{lem:main-technical-lemma}.
    \begin{itemize}
        \item[(a)] Condition (a) follows by property (4) for iteration $i-1$, which is assumed inductively.
        \item[(b)] To establish condition (b), fix a set $C\in\mathcal{R}_{\ge i+1}$. We claim that $\mathbf{b}_0(C)=\mathbf{b}_{i-1}(C)$. Assuming this were the case, then $\mathbf{s}(C)=\mathbf{b}_{i-1}(C)=\mathbf{b}_0(C)=\mathbf{b}(C)\le\delta C$, where the inequality follows since the original flow demand $\mathbf{b}\in\mathbb{R}^V$ satisfies $\mathbf{b}(C)\le\delta C$, establishing property (b).

        To show $\mathbf{b}_0(C)=\mathbf{b}_{i-1}(C)$, observe that it is trivial for $i=1$ so assume $i>1$. For any given $j\in[i-1]$, the set $C$ is the disjoint union of sets $C_1,\ldots,C_{\ell}\in\mathcal{R}_{\ge j+1}$ by \Cref{obs:refinement}. Apply property (3) for iteration $j$ to obtain $(\mathbf{b}_{j-1}-\mathbf{b}_j)(C_k)=0$ for all $k\in[\ell]$. Summing over $k\in[\ell]$, we have $(\mathbf{b}_{j-1}-\mathbf{b}_j)(C)=\sum_{k\in[\ell]}(\mathbf{b}_{j-1}-\mathbf{b}_j)(C_k)=0$ so $\mathbf{b}_{j-1}(C)=\mathbf{b}_j(C)$ for all $j\in[i-1]$. Combining this over all iterations $j\in[i-1]$ gives $\mathbf{b}_0(C)=\cdots=\mathbf{b}_{i-1}(C)$, as desired.
    \end{itemize}
    With the conditions fulfilled, \Cref{lem:main-technical-lemma} outputs a flow $f$, which we set as $f_i$, and a vector $\mathbf{t}\in\mathbb{R}^V$, which we set as $\mathbf{b}_i$. The flow satisfies property (1) due to property ($i$) and satisfies property (2) due to property ($ii$) of the flow $f$. The demand $\mathbf{b}_i$ satisfies property (3) due to property ($iii$) and property (4) due to the second part of property ($i$), completing the proof.
\end{proof}

For the rest of this section, we establish \Cref{lem:main-technical-lemma}. We start with a helper claim about constructing certain demands and flows. Throughout the proof, we say a vector $\mathbf{x}\in\mathbb{R}^{V}$ is supported on $U$ if $\mathbf{x}(v)=0$ for all $v\not\in U$.

\begin{claim}\label{cl:R-to-P-routing}
    For any $i\in[L-1]$ and $\mathbf{x}\in\R^V$ {supported on $V_{i+1}$} with $|\mathbf{x}|\le\deg_{\partial \mathcal{R}_{\ge i}}$, there exists $\mathbf{y}\in\mathbb{R}^V$ such that 
    \begin{enumerate}
        \item $|\mathbf{y}|\le 6\deg_{\partial\overline{\mathcal{P}}_i}+6L\beta\deg_{\partial\overline{\mathcal{P}}_{i+1}}$,
        \item $\mathbf{y}$ is supported on $V_{i+1}$,
        \item for all clusters $C\in\overline{\mathcal{P}}_{i+1}$, we have $(\mathbf{x}-\mathbf{y})(C)=0$, and
        \item there exists a flow routing demand $\mathbf{x}-\mathbf{y}$ with congestion $12L\beta$.
    \end{enumerate}
\end{claim}
\begin{proof}
    Much of the argument mimics the proof of Claim 4.8 of~\cite{DBLP:conf/soda/0006R025}. 
    \begin{subclaim} \label{subcl:between-level-routing}
        For any $\mathbf{s}\in\R^V_{\ge0}$ with $\mathbf{s}(v)\le \deg_{\partial\overline{\mathcal{P}}_{i+1}}(v)$ for $v\in V_{i+1}$ and $\mathbf{s}(v)\le \deg_{\partial V_{i+1}}(v)$ for $v\not\in V_{i+1}$, there exists $\mathbf{t}\in\R^V_{\ge0}$ supported on $V_{i+1}$ with $\mathbf{t}\le \deg_{\partial\overline{\mathcal{P}}_i}/2$ and a flow routing $\mathbf{s}-\mathbf{t}$ with congestion $3\beta$. 
    \end{subclaim}
    \begin{proof}
        By assumption (3) of \Cref{lem:congestion-approximator-guarantee}, there is a flow in $G$ with congestion $\beta$ such that each vertex $v\in V_{i+1}$ sends $\deg_{\partial\overline{\mathcal{P}}_{i+1}}(v)$ flow and receives at most $\frac{1}{4}\deg_{\partial\overline{\mathcal{P}}_i}(v)$ flow. Scaling this flow up by a factor of $2$, there exists a flow sending $2\deg_{\partial\overline{\mathcal{P}}_{i+1}}(v)$ flow and receiving at most $\frac{1}{2}\deg_{\partial\overline{\mathcal{P}}_i}(v)$ flow, for each $v\in V_{i+1}$. We augment this flow by adding full flow on each edge $(u,v)\in\partial V_{i+1}$, which only increases the flow along any edge by a factor of $1\le\beta$ times its capacity. Hence, the total congestion is at most $3\beta$. Note that $\partial\overline{\mathcal{P}}_{i+1}$ includes $\partial V_{i+1}$ as a subset, so augmenting the flow causes {vertices $v\not\in V_{i+1}$ to send $\deg_{\partial V_{i+1}}(v)$ flow, and vertices in $v\in V_{i+1}$ to receive $\deg_{\partial V_{i+1}}(v)\le\deg_{\partial\overline{\mathcal{P}}_{i+1}}(v)$ flow, moving some source without changing any of the sinks.}
        This results in a flow with congestion $3\beta$ sending at least $\deg_{\partial\overline{\mathcal{P}}_{i+1}}(v)$ for each $v\in V_{i+1}$ and $\deg_{\partial V_{i+1}}(v)$ for each $v\not\in V_{i+1}$, and still receiving at most $\frac{1}{2}\deg_{\partial\overline{\mathcal{P}}_i}(v)$ flow for each $v\in V_{i+1}$. Take the path decomposition of the flow where each $v\in V_{i+1}$ is the start of at least $\deg_{\partial\overline{\mathcal{P}}_{i+1}}(v)$ total capacity of flow paths, each $v\not\in V_{i+1}$ is the start of at least $\deg_{\partial V_{i+1}}(v)$ total capacity of flow paths, and each $v\in V_{i+1}$ is the end of at most $\frac{1}{2}\deg_{\partial\overline{\mathcal{P}}_i}(v)$ total capacity of flow paths. Since $\mathbf{s}\le\deg_{\partial\overline{\mathcal{P}}_{i+1}}$ on $V_{i+1}$ and $\mathbf{s}\le \deg_{\partial V_{i+1}}$ on $V\setminus V_{i+1}$, we can remove or decrease the capacity of paths until each $v\in V$ is the start of $\mathbf{s}(v)$ total capacity of paths. Let $\mathbf{t}\in\mathbb{R}^V_{\ge 0}$ be the vector such that each vertex $v\in V$ is the end of $\mathbf{t}(v)$ total capacity of paths, which satisfies $\mathbf{t}\le\deg_{\partial\overline{\mathcal{P}}_{i}}/2$. The resulting flow routes $\mathbf{s}-\mathbf{t}$ with congestion $3\beta$, as desired.        \renewcommand{\qedsymbol}{$\diamond$}
    \end{proof}
    
    \begin{subclaim}\label{subcl:R-to-P-routing-propertyless}
        For any $i\in[L-1]$ and $\mathbf{x}\in\R^V_{\ge0}$ with $\mathbf{x}\le\deg_{\partial\mathcal{R}_{\ge i}}$, there exists $\mathbf{y}\in\mathbb{R}^V_{\ge0}$ with $\mathbf{y}\le 2\deg_{\partial\overline{\mathcal{P}}_i}$ and a flow routing demand $\mathbf{x}-\mathbf{y}$ with congestion $(6L-6i)\beta$. 
    \end{subclaim}
    \begin{proof}
        We prove the statement by induction from $i=L$ down to $i=1$. For the base case $i=L$, we can simply set $\mathbf{y}=\mathbf{x}$ since $\mathcal{R}_{\ge L}=\overline{\mathcal{P}}_L$ so that $\mathbf{x}-\mathbf{y}=0$ is trivially routable.

        For the inductive step, define $\mathbf{x}'\in\mathbb{R}^V_{\ge0}$ from $\mathbf{x}$ as
        \begin{equation}\mathbf{x}'(v)=
            \begin{cases}
                \frac{\deg_{\partial\mathcal{R}_{\ge i+1}}(v)}{\deg_{\partial\mathcal{R}_{\ge i}}(v)}\cdot\mathbf{x}(v)\qquad&\text{if }\deg_{\partial\mathcal{R}_{\ge i}}(v)>0\text{ and}\\
                0&\text{otherwise},
            \end{cases}
        \end{equation}
        which in particular satisfies $\mathbf{x}'\le\deg_{\partial\mathcal{R}_{\ge i+1}}$. This means we can use the inductive hypothesis to route the $\mathbf{x}'$ portion of $\mathbf{x}$. Indeed, by induction, there exists $\mathbf{y}'\in\mathbb{R}^V_{\ge0}$ with $\mathbf{y}'\le2\deg_{\partial\overline{\mathcal{P}}_{i+1}}$ and a flow $f_1$ routing demand $\mathbf{x}'-\mathbf{y}'$ with congestion $(6L-6(i+1))\beta$. 
        
        Let $\mathbf{s}\in\mathbb{R}^{V}$ be the vector where $\mathbf{s}(v)=\mathbf{y}'(v)/2$ for each $v\in V_{i+1}$ and $\mathbf{s}(v)=\min(\deg_{\partial V_{i+1}}(v),\mathbf{y}'(v)/2)$ for $v\not\in V_{i+1}$. Applying \Cref{subcl:between-level-routing} on $\mathbf{s}$, there exists a vector $\mathbf{t}\in\mathbb{R}^{V}_{\ge0}$ supported on $V_{i+1}$ with $\mathbf{t}\le\deg_{\partial\overline{\mathcal{P}}_i}/2$ and a flow routing demand $\mathbf{s}-\mathbf{t}$ with congestion $3\beta$. Scaling this flow up by a factor of 2, we obtain a flow $f_2$ routing demand $2\mathbf{s}-2\mathbf{t}$ with congestion $6\beta$.

        The final flow is the sum $f=f_1+f_2$, which routes demand $(\mathbf{x}'-\mathbf{y}')+(2\mathbf{s}-2\mathbf{t})$ and has congestion $(6L-6(i+1))\beta+6\beta=(6L-6i)\beta$, as desired. We define $\mathbf{y}=\mathbf{x}-\mathbf{x}'+\mathbf{y}'-2\mathbf{s}+2\mathbf{t}$ so that the demand routed by $f$ is exactly $\mathbf{x}-\mathbf{y}$. Note that since $\mathbf{x}'\le\mathbf{x}$ and $2\mathbf{s}\le\mathbf{y}'$, we have $\mathbf{y}\ge 0$, as desired. To complete the induction, it remains to prove that $\mathbf{y}(v)\le 2\deg_{\partial\overline{\mathcal{P}}_i}(v)$ for $v\in V$. We consider the cases where $v\in V_{i+1}$ and $v\not\in V_{i+1}$ separately. 

        If $v\in V_{i+1}$, then $\mathbf{y}'(v)=2\mathbf{s}(v)$ so $\mathbf{y}(v)=(\mathbf{x}-\mathbf{x}'+2\mathbf{t})(v)$. Since $\mathbf{t}(v)\le\deg_{\partial\overline{\mathcal{P}}_i}(v)/2$, it suffices to prove that $\mathbf{x}(v)-\mathbf{x}'(v)\le \deg_{\partial\overline{\mathcal{P}}_i}(v)$. If $\deg_{\partial\mathcal{R}_{\ge i}}(v)=0$, then $\mathbf{x}(v)=\mathbf{x}'(v)=0$, so $(\mathbf{x}-\mathbf{x}')(v)=0\le\deg_{\partial\overline{\mathcal{P}}_i}(v)$ trivially. Otherwise, we have $\deg_{\partial\mathcal{R}_{\ge i}}(v)>0$, so
        \begin{align*}
            (\mathbf{x}-\mathbf{x}')(v)&=\left(\frac{\deg_{\partial\mathcal{R}_{\ge i}}(v)-\deg_{\partial\mathcal{R}_{\ge i+1}}(v)}{\deg_{\partial\mathcal{R}_{\ge i}}(v)}\right)\mathbf{x}(v)\le\frac{\deg_{\partial\overline{\mathcal{P}}_i}(v)}{\deg_{\partial\mathcal{R}_{\ge i}}(v)}\mathbf{x}(v)\le \deg_{\partial\overline{\mathcal{P}}_i}(v),
        \end{align*}
        where the first inequality holds by \Cref{obs:difference-of-R}.

        If $v\not\in V_{i+1}$, then we know $\mathbf{s}(v)=\min(\deg_{\partial V_{i+1}}(v),\mathbf{y}'(v)/2)$ by definition and $\mathbf{t}(v)=0$ since $\mathbf{t}$ is supported on $V_{i+1}$. Furthermore, we have that
        $\deg_{\partial\mathcal{R}_{\ge i}}(v)=\deg_{\partial\mathcal{R}_{\ge i+1}}(v)$ by \Cref{obs:equal-degrees-not-in-V-i+1}, so $\mathbf{x}(v)=\mathbf{x}'(v)$. Combining the above, we have that 
        \begin{align}
            \mathbf{y}(v)=\mathbf{y}'(v)-2\mathbf{s}(v)=\mathbf{y}'(v)-2\min(\deg_{\partial V_{i+1}}(v),\mathbf{y}'(v)/2)=\max(\mathbf{y}'(v)-2\deg_{\partial V_{i+1}}(v),0).\label{eq:bound-y-1}
        \end{align} 
        By definition of $\mathbf{y}'$, we already have $\mathbf{y}'(v)\le2\deg_{\partial\overline{\mathcal{P}}_{i+1}}(v)$. Furthermore, we have already shown in \Cref{obs:equal-degrees-not-in-V-i+1} that for $v\not\in V_{i+1}$, we have $\partial\overline{\mathcal{P}}_{i+1}\setminus\partial V_{i+1}\subseteq \partial\overline{\mathcal{P}}_i$. This implies that 
        \begin{align}\label{eq:bound-y-2}
            \deg_{\partial\overline{\mathcal{P}}_{i+1}}(v)-\deg_{\partial V_{i+1}}(v)=\deg_{\partial\overline{\mathcal{P}}_{i+1}\setminus\partial V_{i+1}}(v)\le\deg_{\partial\overline{\mathcal{P}}_{i}}(v),    
        \end{align}
        where the equality uses $\partial\overline{\mathcal{P}}_{i+1}\supseteq \partial V_{i+1}$. Combining \Cref{eq:bound-y-1,eq:bound-y-2} completes the proof. \renewcommand{\qedsymbol}{$\diamond$}
    \end{proof}
    The remainder of the proof is almost identical to that of Claim 4.8 and Subclaim 4.11 in~\cite{DBLP:conf/soda/0006R025}. The following subclaim almost completes the proof of \Cref{cl:R-to-P-routing}, except $\mathbf{x}$ and $\mathbf{y}$ are restricted to being non-negative.
    \begin{subclaim}\label{subcl:non-negative-version-of-claim}
        For any $i\in[L-1]$, consider any vector $\mathbf{x}\in\mathbb{R}^V_{\ge 0}$ supported on $V_{i+1}$ with $\mathbf{x}\le\deg_{\partial\mathcal{R}_{\ge i}}$. There exists a vector $\mathbf{y}\in\mathbb{R}^V_{\ge 0}$ such that
        \begin{enumerate}
            \item $\mathbf{y}\le6\deg_{\partial\overline{\mathcal{P}}_i}+6L\beta\deg_{\partial\overline{\mathcal{P}}_{i+1}}$,
            \item $\mathbf{y}$ is supported on $V_{i+1}$,
            \item for all clusters $C\in\overline{\mathcal{P}}_{i+1}$, we have $(\mathbf{x}-\mathbf{y})(C)=0$, and
            \item there is a flow routing demand $\mathbf{x}-\mathbf{y}$ with congestion $6L\beta$.
        \end{enumerate}
    \end{subclaim}
    \begin{proof}
        Apply \Cref{subcl:R-to-P-routing-propertyless} on vector $\mathbf{x}$ to obtain a vector $\mathbf{y}'\in\mathbb{R}^V_{\ge0}$ and a flow $f$ routing demand $\mathbf{x}-\mathbf{y}'$ with congestion $(6L-6i)\beta\le 6L\beta$. Take a path decomposition of $f$ where each vertex $v$ is the start of $\mathbf{x}(v)$ total capacity of (potentially empty) flow paths and the end of $\mathbf{y}'(v)$ total capacity of (potentially empty) flow paths. For each path starting at a vertex $v$ in some cluster $C\in\overline{\mathcal{P}}_{i+1}$, perform the following operation. If the path contains an edge $(u,w)\in\partial C$ with $u\in C$, replace the path with its prefix ending at $u$; otherwise, do nothing with the path. These modified paths form a new flow $f'$, which also has congestion $6L\beta$. Note that the modified path ends in the same cluster as its starting point. 

        We now bound the difference in the demands routed by $f$ and $f'$. To do so, we consider the difference in the new and old path decompositions. Each vertex $u\in V$ was initially the endpoint of $\mathbf{y}'(u)$ total capacity of paths. We claim that for each cluster $C\in\overline{\mathcal{P}}_{i+1}$, each vertex $u\in C$ becomes the new endpoint of at most $6L\beta\deg_{\partial C}(u)=6L\beta\deg_{\partial\overline{\mathcal{P}}_{i+1}}(u)$ total capacity of paths. This is because each new endpoint is the result of an edge $(u,w)\in\partial C$ in some path, and the total capacity of such paths is at most $6L\beta\deg_{\partial C}(u)$, since the congestion of $f$ is $6L\beta$. It follows that each vertex $u\in V$ is the endpoint of at most $\mathbf{y}'(u)+6L\beta\deg_{\partial\overline{\mathcal{P}}_{i+1}}(u)$ total capacity of paths in the new flow $f'$.

        Define vector $\mathbf{y}\in\mathbb{R}^V_{\ge0}$ such that each vertex $u\in V$ is the endpoint of $\mathbf{y}(u)$ total capacity of flow paths in the new flow $f'$. In other words, the new flow $f'$ routes demand $\mathbf{x}-\mathbf{y}$, as desired. As previously mentioned, since $f'$ is a truncation of $f$, it has congestion $6L \beta$, giving property (4). We have shown that $\mathbf{y}\le\mathbf{y}'+6L\beta\deg_{\partial\overline{\mathcal{P}}_{i+1}}$, which, combined with the fact that $\mathbf{y}'\le 2\deg_{\partial\overline{\mathcal{P}}_i}$ from \Cref{subcl:R-to-P-routing-propertyless}, gives (1). Since $\mathbf{x}$ is supported on $V_{i+1}$ and we constructed $f'$ so that each flow path starts and ends in the same cluster $C\in\overline{\mathcal{P}}_{i+1}$, all the flow paths in $f'$ also lie in $V_{i+1}$. (Clusters in $\overline{\mathcal{P}}_{i+1}$ are either contained in $V_{i+1}$ or $V \setminus V_{i+1}$, by construction.) Hence, the residual demand $\mathbf{y}$ unrouted by $f'$ also is supported on $V_{i+1}$, yielding (2). Finally, again by our path truncation step, we have $(\mathbf{x}-\mathbf{y})(C)=0$ for all $C\in\overline{\mathcal{P}}_{i+1}$. This yields (3) and completes the proof.
        \renewcommand{\qedsymbol}{$\diamond$}
    \end{proof}

    Finally, we prove~\Cref{cl:R-to-P-routing} using \Cref{subcl:non-negative-version-of-claim}. Given a vector $\mathbf{x}\in\mathbb{R}^V$ with $|\mathbf{x}|\le\deg_{\partial\mathcal{R}_{\ge i}}$, let $\mathbf{x}^+,\mathbf{x}^-\in\mathbb{R}^V_{\ge0}$ be the positive and negative parts of $\mathbf{x}$, so that $\mathbf{x}^+-\mathbf{x}^-=\mathbf{x}$. We apply \Cref{subcl:non-negative-version-of-claim} on $\mathbf{x}^+$ and $\mathbf{x}^-$ separately to obtain $\mathbf{y}^+$ and $\mathbf{y}^-$, respectively, and set $\mathbf{y}=\mathbf{y}^+-\mathbf{y}^-$. The four properties are satisfied immediately by the corresponding four properties in \Cref{subcl:non-negative-version-of-claim}; note that the congestion is now $12L\beta$ because we take the difference of the two flows routing demand $\mathbf{x}^+-\mathbf{y}^+$ and $\mathbf{x}^--\mathbf{y}^-$. 
\end{proof}

Our expander decomposition constructed at level $i$ will be on the set $V_{i+1}\subseteq V$. That is, we will construct a partition $\overline{\mathcal{P}}_{i+1}$ consisting of (1) a partition of $V_{i+1}$, which will be the expander decomposition and (2) a partition of $V\setminus V_{i+1}$, which will be an induced partition from $\overline{\mathcal{P}}_i$.

We now prove \Cref{lem:main-technical-lemma}, restated below, using~\Cref{cl:R-to-P-routing}.

\techlem*
\begin{proof}
    We first construct vector $\mathbf{t}\in\mathbb{R}^V$ as follows. For each set $C\in\mathcal{R}_{\ge i+1}$ such that $C\subseteq V_{i+1}$, define $\mathbf{t}(v)=\mathbf{s}(C)\cdot\deg_{\partial\mathcal{R}_{\ge i+1}}(v)/\delta C$ for all $v\in C$. For each set $C\in\mathcal{R}_{\ge i+1}$ such that $C\subseteq V\setminus V_{i+1}$, define $\mathbf{t}(v)=\mathbf{s}(v)$. Note that each set $C\in\mathcal{R}_{\ge i+1}$ falls into one of the two cases since $\overline{\mathcal{P}}_{i+1}$ consists of a partition of $V_{i+1}$ and a partition of $V\setminus V_{i+1}$, so this fully defines $\mathbf{t}$. 
    
    To see that $\mathbf{t}$ satisfies the second part of property $(i)$, consider $v\in V$. If $v\in V_{i+1}$, then $\mathbf{t}(v)=\mathbf{s}(C)\cdot\deg_{\partial\mathcal{R}_{\ge i+1}}(v)/\delta C$ which satisfies $|\mathbf{t}(v)|\le \deg_{\partial\mathcal{R}_{\ge i+1}}(v)$ by condition (b). If $v\in V\setminus V_{i+1}$, then $\mathbf{t}(v)=\mathbf{s}(v)$ satisfies the second part of property $(i)$ by assumption (a) and the fact that $\deg_{\partial\mathcal{R}_{\ge i}}(v)=\deg_{\partial\mathcal{R}_{\ge i+1}}(v)$ for $v\not\in V_{i+1}$ by \Cref{obs:equal-degrees-not-in-V-i+1}. For property $(iii)$, consider some set $C\in\mathcal{R}_{\ge i + 1}$. If $C\subseteq V_{i+1}$, then $$\mathbf{t}(C)=\sum_{v\in C}\mathbf{t}(v)=\sum_{v\in C}\mathbf{s}(C)\cdot\frac{\deg_{\partial\mathcal{R}_{\ge i+1}}(v)}{\delta C}=\mathbf{s}(C).$$ If $C\subseteq V\setminus V_{i+1}$, then we have $\mathbf{t}(C)=\mathbf{s}(C)$ trivially because $\mathbf{t}(v)=\mathbf{s}(v)$ for each $v\in C$.

    Next, we define the flow $f$. Observe that demand $\mathbf{s}-\mathbf{t}$ is supported on $V_{i+1}$ and furthermore, we have $$|\mathbf{s}-\mathbf{t}|\le |\mathbf{s}|+|\mathbf{t}|\le\deg_{\partial\mathcal{R}_{\ge i}}+\deg_{\partial\mathcal{R}_{\ge i+1}}\le 2\deg_{\partial\mathcal{R}_{\ge i}}.$$ Thus, we can apply \Cref{cl:R-to-P-routing} with $\mathbf{x}=\frac{1}{2}(\mathbf{s}-\mathbf{t})$ to obtain a vector $\mathbf{y}\in\mathbb{R}^V_{\ge 0}$ supported on $V_{i+1}$ such that $|\mathbf{y}|\le 6\deg_{\partial\overline{\mathcal{P}}_i}+6L\beta\deg_{\partial\overline{\mathcal{P}}_{i+1}}$ and $(\frac{1}{2}(\mathbf{s}-\mathbf{t})-\mathbf{y})(C)=0$ for all clusters $C\in\overline{\mathcal{P}}_{i+1}$. Furthermore, we have a flow $f_1$ routing $\frac{1}{2}(\mathbf{s}-\mathbf{t})-\mathbf{y}$ with congestion $12L\beta$, so $2f_1$ routes $\mathbf{s}-\mathbf{t}-2\mathbf{y}$ with congestion $24L\beta$.

    It remains to route $2\mathbf{y}$. Consider a cluster $C\in\overline{\mathcal{P}}_{i+1}$. Since $(\frac{1}{2}(\mathbf{s}-\mathbf{t})-\mathbf{y})(C)=0$ and $(\mathbf{s}-\mathbf{t})(C)=0$, we have $\mathbf{y}(C)=0$ as well. Moreover, for all vertices $v\in C$, we have
    $$|\mathbf{y}(v)|\le 6\deg_{\partial\overline{\mathcal{P}}_i}(v)+6L\beta\deg_{\partial C}(v)\le 12L\beta\deg_{\partial\overline{\mathcal{P}}_i\cup\partial C}(v).$$
    Thus, the scaled down demand $\frac{1}{12L\beta}\mathbf{y}|_C$ satisfies $\left|\frac{1}{12L\beta}\mathbf{y}|_C\right|\le \deg_{\partial\overline{\mathcal{P}}_i\cup\partial C}|_C$. By assumption (2) of \Cref{lem:congestion-approximator-guarantee}, the collection of vertex weightings $\{\deg_{\partial\overline{\mathcal{P}}_i\cup\partial C}|_C\in\mathbb{R}^V_{\ge0}:C\in\overline{\mathcal{P}}_{i+1},C\subseteq V_{i+1}\}$ mixes simultaneously with congestion $\alpha$, so there is a flow $f_2$ routing demand $$\sum_{C\in\overline{\mathcal{P}}_{i+1},C\subseteq V_{i+1}}\frac{1}{12L\beta}\mathbf{y}|_C=\sum_{C\in\overline{\mathcal{P}}_{i+1}}\frac{1}{12L\beta}\mathbf{y}|_C=\frac{1}{12L\beta}\mathbf{y},$$ where the first equality follows since $\mathbf{y}$ is supported on $V_{i+1}$. Thus, we have that $24L\beta f_2$ routes demand $2\mathbf{y}$ with congestion $24L\alpha\beta$.

    The final flow is $f=2f_1+24L\beta f_2$, which routes demand $(\mathbf{s}-\mathbf{t}-2\mathbf{y})+2\mathbf{y}=\mathbf{s}-\mathbf{t}$ and has congestion $24L\beta+24L\alpha\beta\le 48L\alpha\beta$, concluding the proof.
\end{proof}

\subsection{Running Time}

For some of our applications, we will actually need to algorithmically find the routing of the demand, rather than just establishing that such a routing exists. Our proof above already implicitly gives an algorithm for finding the routing. In this subsection, we analyze the runtime. Formally, we will assume two oracles for finding a routing, corresponding to properties (2) and (3) of~\Cref{lem:congestion-approximator-guarantee}. For a demand $\mathbf{b}\in\mathbb{R}^V$ satisfying $|\mathbf{b}|\le \sum_{C\in\overline{\mathcal{P}}_{i+1},C\subseteq V_{i+1}}\deg_{\partial\overline{\mathcal{P}}_{i}\cup\partial C}|_C$ and $\bb(C)=0$ for each $C\in\overline{\mathcal{P}}_{i+1},C\subseteq V_{i+1}$, assume we can find a flow routing $\mathbf{b}$ in $\mathcal{T}_2^i$ time for each $i\in[L-1]$. Let $\mathcal{T}_2=\max_{i\in[L-1]}\mathcal{T}_2^i$. For a demand $\mathbf{b}\in\mathbb{R}^V$ satisfying $\mathbf{b}(v)=0$ for $v\not\in V_{i+1}$ and $-\deg_{\partial\overline{\mathcal{P}}_{i}}(v)/4\le\mathbf{b}(v)\le \deg_{\partial\overline{\mathcal{P}}_{i+1}}(v)$, assume we can find a flow routing the demand in $\mathcal{T}_3^i$ time. Let $\mathcal{T}_3=\max_{i\in[L-1]}\mathcal{T}_3^i$.

Recall that our flow is constructed over $L-1$ iterations. At each iteration, we need to obtain the routing guaranteed by \Cref{lem:main-technical-lemma} once. To obtain the routing from \Cref{lem:main-technical-lemma}, we need to obtain the routing guaranteed by \Cref{cl:R-to-P-routing}, find a routing using oracle $\mathcal{T}_2$, and appropriately modify the routings in $O(m\log(nW))$ additional time.\footnote{Throughout the algorithm, we use the following subroutine: take a path decomposition of a flow vector $f$ and rescale the flow along flow paths with some specified set of endpoints and scalings. This can be implemented using link-cut trees in $O(m\log(nW))$ time, using the transcript of the algorithm from computing the original path decomposition~\cite{DBLP:journals/jcss/SleatorT83}.} To obtain the routing from \Cref{cl:R-to-P-routing}, we again have $L$ inductive steps, each of which uses one oracle call to $\mathcal{T}_3$ and $O(m\log(nW))$ additional time. In total, the runtime is $O(L\cdot\mathcal{T}_2+L^2\mathcal{T}_3+mL^2\log(nW))$.

In our algorithm, we will have the property that $\mathcal{T}_{2}=O(m\log^3(nW))$ and $\mathcal{T}_3=O(m\log(nW))$. We will justify this later in \Cref{sec:combining-everything} but first state our lemma which we use later assuming these parameters.
\begin{lemma}\label{lem:routing-runtime}
    Let $\mathcal{C}$ be as defined in \Cref{lem:congestion-approximator-guarantee} and let $\mathbf{b}$ be a demand such that $|\mathbf{b}(C)|\le\delta C$ for each $C\in\mathcal{C}$. Assuming $\mathcal{T}_2=O(m\log^3(nW))$ and $\mathcal{T}_3=O(m\log(nW))$, we can obtain a flow routing $\mathbf{b}-\mathbf{b}'$ such that $|\mathbf{b}'|\le\deg_{\partial\overline{\mathcal{P}}_L}$  with congestion $48\alpha \beta L^2$ and in time $O(m\log^4(nW))$.
\end{lemma}

\section{Building Our Congestion-Approximator}\label{sec:bottom-up-construction}

The partitioning algorithm starts with the partition $\overline{\mathcal{P}}_1= \{\{v\} : v \in V \}$ of singleton clusters. The algorithm then iteratively constructs partition $\overline{\mathcal{P}}_{i+1}$ given the current partitions $\overline{\mathcal{P}}_1,\ldots,\overline{\mathcal{P}}_i$. The lemma below establishes this iterative algorithm, where we substitute $L$ for $i$.

\begin{theorem}\label{thm:iterative-algorithm}
    Consider a capacitated graph $G=(V,E,\bc)$. Suppose there exists partitions $\overline{\mathcal{P}}_1,\ldots,\overline{\mathcal{P}}_L$ that satisfy the following properties:
        \begin{enumerate}
            \item $\overline{\mathcal{P}}_1$ is the partition $\{\{v\}:v\in V\}$ of singleton clusters.
            \item For each $i\in[L-1]$, the collection of vertex weightings $\{\deg_{\partial\overline{\mathcal{P}}_i\cup\partial C}|_C\in\mathbb{R}_{\ge0}^V:C\in\overline{\mathcal{P}}_{i+1},C\subseteq V_{i+1}\}$ mixes simultaneously in $G$ with congestion $\alpha=O(\log^3(nW))$.
            \item For each $i\in[L-1]$, there is a flow in $G$ with congestion $\beta=O(1)$ such that each $v\in V_{i+1}$ sends $\deg_{\partial\overline{\mathcal{P}}_{i+1}}(v)$ flow and each $v\in V_{i+1}$ receives at most $\frac{1}{4}\deg_{\partial\overline{\mathcal{P}}_i}(v)$ flow.
            \item For each $i\in[L-1]$, the size of the boundaries are decreasing: $\delta\overline{\mathcal{P}}_{i+1}\le \delta\overline{\mathcal{P}}_{i}/2$.
        \end{enumerate}
    Then, there is an algorithm running in {${O}(m\log^{8}(nW))$} time that constructs a partition $\overline{\mathcal{P}}_{L+1}$ such that properties (2), (3), and (4) hold for $i=L$ as well.
\end{theorem}

Note that the first three properties of~\Cref{thm:iterative-algorithm} are the same as those in~\Cref{lem:congestion-approximator-guarantee}. As suggested in \Cref{lem:congestion-approximator-guarantee}, the new partitions will be constructed by finding a partition $\mathcal{P}_{L+1}$ of a subset $V_{L+1}\subseteq V$ and combining it with the induced partition from $\overline{\mathcal{P}}_{L}$ on the remaining vertices $V\setminus V_{L+1}$. This will be done via our faster algorithm for weak expander decompositions from \Cref{sec:weak-expander-decomp}. Before describing our algorithm further, we first show that $L=O(\log(nW))$ iterations suffice to obtain a congestion-approximator.

\begin{corollary} \label{cor:main-flow-result}
    Let $G = (V,E,\bc)$ be a capacitated graph with $\bc \in [1,W] \cap \Z$. There is an $O(m\log^{9}(nW)\log\log(nW))$ time algorithm to construct a congestion-approximator $\mathcal{C}$ of $G$ with quality $O(\log^5(nW))$. This implies an $O(m\log^9(nW)\log\log(nW)+m\log^6(nW)/\epsilon)$ time algorithm for $(1-\epsilon)$-approximate max flow.
\end{corollary}
\begin{proof}
    Recall that we have $\delta\overline{\mathcal{P}}_{i+1}\le\delta\overline{\mathcal{P}}_i/2$ for each $i\in[L]$. This ensures that for $L=O(\log(nW))$, we have $\delta\overline{\mathcal{P}}_L<1$. Since all edge capacities are assumed to be integral, this implies that $\delta\overline{\mathcal{P}}_L=0$ so we must have $\overline{\mathcal{P}}_L=\{V\}$, fulfilling property (1) of \Cref{thm:congestion-approximator-guarantee}. By construction, the partitions $\overline{\mathcal{P}}_1,\ldots,\overline{\mathcal{P}}_L$ also satisfy properties (2) and (3) with $\alpha=O(\log^3(nW))$ and $\beta=O(1)$, so \Cref{thm:congestion-approximator-guarantee} implies that $\mathcal{C}$ is a congestion-approximator with quality $16\alpha\beta L^2=O(\log^5(nW))$. The runtime follows since we apply \Cref{thm:iterative-algorithm} iteratively $L$ times, the total runtime is $O(m\log^{9}(nW)\log\log(nW))$. The claim about approximate max flow then follows by \cite{sherman2017area,jambulapati2023revisiting}.
\end{proof}

In the remainder of the section, we prove \Cref{thm:iterative-algorithm}. To do so, we apply our weak expander decomposition algorithm from~\Cref{sec:weak-expander-decomp}, which requires us to implement~\Cref{ora: matching player,ora: grafting flow} of~\Cref{thm: weak-decomp-del2}. An approximate max flow oracle suffices to implement both of these oracles, but that is exactly the problem we are trying to solve. To resolve this, \cite{DBLP:conf/soda/0006R025} observed that the pseudo-congestion-approximator is a real congestion-approximator on some modified graph. They then use the real congestion-approximator on the modified graph to obtain the required flow oracles via additional post-processing, at the loss of additional log factors. Instead, we directly show that the pseudo-congestion-approximator made up of the partitions $\overline{\mathcal{P}}_1,\ldots,\overline{\mathcal{P}}_L$ suffices for solving the approximate max flow instances required by the cut-matching game, with no log factor loss. This suffices for achieving properties (2) and (3) of~\Cref{thm:iterative-algorithm} for the next iteration of our partitioning algorithm, as proven already in \Cref{sec:weak-expander-decomp}.

\subsection{Cut-Matching via a Pseudo-Congestion-Approximator}\label{sec:prop-2}

We apply our weak expander decomposition algorithm from \Cref{sec:weak-expander-decomp}. The main difficulty here is developing an efficient algorithm for the matching player (i.e., implementing~\Cref{ora: matching player}), which requires solving a max flow problem (approximately). Previous work~\cite{sherman2017area} shows how to convert a congestion-approximator into an approximate max flow algorithm. In this section, we build on~\cite{sherman2017area} to show that our pseudo-congestion-approximator is sufficient to approximately solve our specific max flow instance arising in the cut-matching game.

We first recall the flow instance which we need to solve for the matching player. We start with a demand $\mathbf{d} = \deg_{\partial \overline{P}_L}$ and the starting partition $\mathcal{A}_0^\circ=\{V\}$. At each iteration $t$ of the cut-matching game, we will maintain a collection $\mathcal{A}_t^\circ$ of disjoint subsets of $V$. From the cut player, we obtain sets $L_A\sqcup R_A=A$ for each $A\in\mathcal{A}_t^\circ$. We wish to solve the following flow problem guaranteeing the properties of~\Cref{ora: matching player}. Let $G_t$ be the graph of $G$ with all edges between components in $\mathcal{A}_t^\circ$ removed and all edge capacities scaled up by a factor of $2/\phi$. Add a source vertex $s$ and a sink vertex $t$. For $A\in\mathcal{A}_t^\circ$, we do the following: for each $u\in L_A$, we add an edge $(s,u)$ with capacity $\mathbf{d}_{t-1}(u)$, and, for each $v\in R_A$, we add an edge $(v,t)$ with capacity $\mathbf{d}_{t-1}(v)$.

To solve this flow problem, We use the following instantiation of Sherman's algorithm~\cite{sherman2017area} as stated in~\cite{DBLP:conf/ipco/LiL25}, but with the running time speedup from~\cite{jambulapati2023revisiting} (see Section~5 of their arXiv version), which partially routes the demand, leaving a small amount of residual demand.

\begin{lemma}[Almost-Route]\label{lem:almost-route}
Consider a graph $G=(V,E,\bc)$, two vertices $s,t\in V$, parameters $\epsilon,\tau>0$, and a laminar family of vertex subsets $\mathcal C$. There is an $O(m\log(n)/\epsilon)$ time algorithm that computes either
\begin{enumerate}
 \item An $(s,t)$-cut in $G$ of value less than $\tau$, or
 \item A flow $f$ in $G$ routing a demand $\mathbf d$ such that the residual demand $\tilde{\mathbf d}=\tau(\mathbf 1_s-\mathbf 1_t)-\mathbf d$ satisfies \\$|\tilde{\mathbf d}(C)|\le\epsilon \cdot \delta C$ for all $C\in\mathcal C$.
 \end{enumerate}
\end{lemma} 

In the setting where $\mathcal{C}$ is a congestion-approximator, the residual demand can be routed via the congestion-approximator so that the flow satisfies the input demand. We will show that in our case, a pseudo-congestion-approximator also suffices to route the residual demand.

\begin{lemma}\label{lem:oracle1}
    Given a laminar family of sets $\mathcal{C}$ defined as in \Cref{lem:congestion-approximator-guarantee} for $G=(V,E, \bc)$, there is an algorithm running in $O(m\log^6(nW)\log\log(nW))$ time which satisfies the guarantees of~\Cref{ora: matching player}.
\end{lemma}
\begin{proof}    
    Let  $\mathcal{C}_t=\mathcal{C}\cup\{\{s\},\{t\}\}$. Note that $\mathcal{C}_t$ is a laminar family of subsets of vertices in $G_t$. We apply \Cref{lem:almost-route} on the flow problem on $G_t$ with laminar family $\mathcal{C}_t$, with  parameter $\epsilon=\floweps/(4320\alpha\beta L^2)$ for $\floweps=\Theta(1/\log^2(nW))$ and $\tau$ ranging from $\floweps \mathbf{d}(V)$ to $\mathbf{d}_{t-1}(V)$. We choose our $\tau$ via binary search up to error $\floweps \mathbf{d}(V)/2$: if we find a flow, we increase $\tau$ and if we find a cut, we decrease $\tau$. Thus, we apply \Cref{lem:almost-route} a total of $O(\log(1/\floweps^2))=O(\log\log(nW))$ times, giving a runtime of $O(m\log^6(nW)\log\log(nW))$.
    
    At the end of the binary search, we will find some $\tau$ where we find a flow $f_1$ of value $\tau$ but there is no flow of value $\tau+\floweps\mathbf{d}(V)$, as certified by an $(s,t)$-cut $(S\cup\{s\},V\cup\{t\}\setminus S)$ of value less than $\tau+\floweps\mathbf{d}(V)$. By property (2) of \Cref{lem:almost-route}, if we find a flow routing the residual demand $\tilde{\mathbf{d}}$, the combined flow routes $\tau(\mathbf{1}_s-\mathbf{1}_t)$, as desired. As a result, we turn to finding a flow routing $\tilde{\mathbf{d}}$. We will not be able to do this exactly, but we will still construct some flow $f'$ routing $\tau'(\mathbf{1}_s-\mathbf{1}_t)$ for some $\tau'\ge\tau-\floweps\mathbf{d}(V)$. Note that if $\tau\le \floweps\mathbf{d}(V)$, we just take $f'$ as the empty flow.
    
    Recall that we have the guarantee that $|\tilde{\mathbf{d}}(C)|\le\epsilon\cdot\delta_{G_t} C$ for $C\in\mathcal{C}_t$. In particular, this implies that $|\tilde{\mathbf{d}}(\{s\})|\le \epsilon\cdot\sum_{A\in\mathcal{A}_t^\circ}\deg_{\partial\overline{\mathcal{P}}_L}(A^l)$ and $|\tilde{\mathbf{d}}(\{t\})|\le\epsilon\cdot\sum_{A\in\mathcal{A}_t^\circ}\deg_{\partial\overline{\mathcal{P}}_L}(A^r)$. We first send the flow from $s$ and $t$ to the nodes in $\cup_{A\in\mathcal{A}_t^\circ}L_A$ and $\cup_{A\in\mathcal{A}_t^\circ}R_A$, proportional to the induced degrees on $\partial\overline{\mathcal{P}}_L$, such that there is no residual left on $s$ or $t$. Let $f_2$ denote this flow. Due to our bound on the residual on $\{s\}$ and $\{t\}$, sending this flow increases the residual on each vertex by at most $\epsilon\cdot \deg_{\partial\overline{\mathcal{P}}_L}(v)$.

    Next, we wish to route the residual of $f_1+f_2$, which lies entirely on $V$. Let $\mathbf{b}$ denote this residual. We have the guarantee that $|\mathbf{b}|\le|\tilde{\mathbf{d}}|+\epsilon\cdot\deg_{\partial\overline{\mathcal{P}}_L}$. Since $|\tilde{\mathbf{d}}(C)|\le \epsilon\cdot \delta_{G_t}(C)\le\epsilon\cdot\delta_{G}(C)+\epsilon\cdot\deg_{\partial\overline{\mathcal{P}}_L}(C)$, this implies that $|\mathbf{b}(C)|\le \epsilon\cdot\delta_{G}(C)+2\epsilon\cdot\deg_{\partial\overline{\mathcal{P}}_L}(C)$. Recall that $\mathcal{C}$ is a refinement of $\overline{\mathcal{P}}_L$, so every set $C\in\mathcal{C}$ is completely contained in a component of $\overline{\mathcal{P}}_L$. Consequently, we have $\deg_{\partial\overline{\mathcal{P}}_L}(C)\le \delta_G(C)$, which implies that $|\mathbf{b}(C)|\le 3\epsilon\cdot\delta_G(C)$.

    It remains to route demand $\mathbf{b}$. To do this, we apply the routing guaranteed by the pseudo-congestion-approximator. By \Cref{lem:routing-runtime}, there exists a flow $f_3$ in $G$ with congestion $3\epsilon\cdot 48\alpha\beta L^2 = \floweps/30$ routing $\mathbf{b}-\mathbf{b}'$ for $|\mathbf{b}'|\le3\epsilon\cdot\deg_{\partial\overline{\mathcal{P}}_L}$ which we can find in ${O}(m\log^6(nW))$ time. Note that this same flow would have congestion $\phi\gamma/60$ in $G_t$, since the edges are scaled up by a $2/\phi$ factor (except that the flow $f_3$ uses some edges not in $G_t$). Now, we have a flow $f_1+f_2+f_3$ which routes the demand $\tau(\mathbf{1}_s-\mathbf{1}_t)-\mathbf{b}'$ such that the residual demand is at most $|\mathbf{b}'|\le 3\epsilon\cdot\deg_{\partial\overline{\mathcal{P}}_L}$. 
    
    Observe that this flow is not immediately a flow in $G_t$, since $f_3$ may use intercomponent edges in $G$, not in $G_t$. We perform a path decomposition of this flow, so that the start and end of each flow path is an endpoint of some edge in $\partial\overline{\mathcal{P}}_L$. To obtain our final flow $f'$, we first remove every flow path that crosses the boundary $\partial_G A$ of some $A\in\mathcal{A}_t^\circ$ or ends at some node $v\in V$ which is not a source or sink in any flow instance. Finally, we scale down the flow by a factor of $1+\floweps/25$, and denote the resulting flow $f'$. It now remains to show that $f'$ satisfies our desired properties.

    First, we wish to show that $f'$ is a valid flow from $s$ to $t$ in $G_t$. By construction, we removed all flow paths that crosses the boundary $\partial_GA$ of some $A\in\mathcal{A}_t^\circ$, so the flow is indeed on $G_t$. For the congestion, recall that $f_1$ had congestion $1$, $f_2$ had congestion $\epsilon$, and $f_3$ had congestion $\floweps/30$. In total, the congestion is upper bounded by $1+\floweps/25$, so the resulting flow $f'$ has congestion at most $1$. Finally, we remove all flow paths in the path decomposition that have an endpoint at some node other than $s$ or $t$, so it is an $s$-$t$ flow.

    Next, we want to show that the total volume of flow removed in defining $f'$ from $f_1+f_2+f_3$ is small. Recall that $f_1$ routes $\tau(\mathbf{1}_s-\mathbf{1}_t)-\tilde{\mathbf{d}}$ and $f_2+f_3$ routes $\tilde{\mathbf{d}}-\mathbf{r}$, where $|\mathbf{r}|\le 3\epsilon\cdot\deg_{\partial\overline{\mathcal{P}}_L}$, so $f_1+f_2+f_3$ routes $\tau(\mathbf{1}_s-\mathbf{1}_t)-\mathbf{r}$. The volume of flow paths removed due to having an endpoint at a node other than $s$ and $t$ is upper bounded by the total residual demand $\mathbf{r}$, which is upper bounded by $3\epsilon\cdot\delta\overline{\mathcal{P}}_L$. To bound the volume of flow paths removed due to the flow path crossing some boundary edge in $\partial_GA$, first observe that all such flow paths must come from $f_3$, since $f_1$ and $f_2$ are both supported on $G_t$. But the volume of flow paths in $f_3$ is at most $(\floweps/30)\cdot\sum_{A\in\mathcal{A}_t^\circ}\delta_G A$, since the congestion of $f_3$ is at most $\floweps/30$. We know that the boundary of the decomposition $\mathcal{A}_t^\circ$ is at most $0.5\cdot\delta\overline{\mathcal{P}}_L$, which implies that the total volume of flow paths removed in total is at most $(\floweps/15)\cdot\delta\overline{\mathcal{P}}g_L\le(\floweps/15)\cdot\mathbf{d}(V)$. Since we scale the flow down by $1+\floweps/25$, this means that the resulting flow $f'$ sends at least 
    \[
    (\tau-\floweps\mathbf{d}(V)/15)\cdot\frac{1}{ 1+ \floweps/25} \geq (\tau-\floweps\mathbf{d}(V)/15)\cdot (1-\floweps/25)\ge \tau-\floweps\mathbf{d}(V)
    \]
    flow from $s$ to $t$, with the second inequality using that $\tau \leq \bd(V)$.

    Finally, we define the output of the oracle. Recall that we have an $(s,t)$-cut $(S\cup\{s\},V\cup\{t\}\setminus S)$ of value less than $\tau+\floweps\mathbf{d}(V)$. For each $A\in\mathcal{A}_t^{\circ}$, define $C_A\subseteq A$ as $A\cap S$. First, we wish to show that these cuts are sufficiently sparse. Note that we have $\sum_A \Delta(C_A)\ge \tau-\sum_A\Delta(A\setminus C_A)-\floweps \mathbf{d}(V)$ since $f'$ sends at least $\tau-\floweps\mathbf{d}(V)$ flow from $s$ to $t$. We additionally have that $\tau+\floweps\mathbf{d}(V)\ge \sum_{A}(\delta_{G_t}(C_A,A\setminus C_A)+\Delta(A\setminus C_A)+\nabla(C_A))$ since the right hand side is the value of the cut $(S\cup\{s\},V\cup\{t\}\setminus S)$, and we know the cut value is at most $\tau+\floweps\mathbf{d}(V)$. Combining these together implies that $\sum_A\Delta(C_A)\ge \sum_A\delta_{G_t}(C_A,A\setminus C_A)+\sum_A\nabla(C_A)-2\floweps\mathbf{d}(V)$. Since the capacities in $G_t$ are scaled up $2/\phi$ from those in $G$, this in particular implies that $\frac{2}{\phi}\sum_A \delta_G(C_A,A\setminus C_A)\le \sum_{A}\mathbf{d}_{t-1}(C_A)+2\floweps\mathbf{d}(V)$, as required by property (1) of the oracle. 

    Next, we wish to lower bound the amount of source routed in the remaining graph $\bigcup_A (A\setminus C_A)$.
    Since $f'$ sends at least $\tau-\floweps\mathbf{d}(V)$ through the cut and the cut has capacity at most $\tau+\floweps\mathbf{d}(V)$, there is at most $2\floweps\mathbf{d}(V)$ capacity of cut edges which are unsaturated by $f'$. Initialize $\mathcal{A}_t''=\mathcal{A}_t^{\circ}$. For each $A\in\mathcal{A}_t''$, if $f'$ does not send more than $\Delta(A\setminus C_A)-20000\floweps\mathbf{d}(A)$ source through the cut $(C_A,A\setminus C_A)$, then remove $A$ from $\mathcal{A}''_t$.
    Observe that if $A$ is removed, this implies that $20000\floweps\mathbf{d}(A)$ volume of edges from $s$ to $A\setminus C_A$ is not saturated. Consequently, if the volume of removed $A$ exceeds $0.001\mathbf{d}_{t-1}(V)$, this would contradict the fact that there are at most $2\floweps\mathbf{d}(V)$ capacity of cut edges which are unsaturated. This uses that we can choose $\floweps$ to ensure that $\bd_{t-1}(V) \geq \bd(V)/10$. Hence, $\mathbf{d}_{t-1}(\bigcup_{A\in\calA''_{t}}A)\ge0.999 \mathbf{d}_{t-1}(\bigcup_{A\in\calA^\circ_{t}}A)$, and each $A\in\mathcal{A}''_{t}$ satisfies property (2) of the oracle. 

    Lastly, we wish to prove that $\mathbf{d}_{t-1}(C_A)\le \mathbf{d}_{t-1}(A)/2$. Since $\mathbf{d}_{t-1}(C_A)=\Delta(C_A)+\nabla(C_A)$, it suffices to bound these terms separately. We have $\Delta(C_A)\le \mathbf{d}(A)/8$ by construction of the cut $A=L_A\sqcup R_A$. To bound the sink, recall that $\sum_{A\in\mathcal{A}_{t-1}}\nabla(C_A)\le \tau+\gamma\mathbf{d}(V)\le \mathbf{d}_{t-1}(V)/7$ because each sink in $C_A$ corresponds to a cut edge and $\tau\le \mathbf{d}_{t-1}(V)/8$. Now, initialize $\calA'_t=\calA''_t$. For each $A\in\calA''_t$, remove $A$ if $\nabla(C_A)\ge \mathbf{d}_{t-1}(A)/3$. Suppose the volume of removed $A$ in defining $\mathcal{A}'_t$ from $\mathcal{A}''_t$ is more than $0.43\mathbf{d}_{t-1}(V)$. This would imply that the total sink in $C_A$ in all the removed components $A\in\mathcal{A}''_t\setminus\calA'_t$ is more than $0.43\cdot \mathbf{d}_{t-1}(V)/3>\mathbf{d}_{t-1}(V)/7$. But this is impossible, since we previously bounded the total sink as $\sum_{A\in\mathcal{A}_t}\nabla(C_A)\le \mathbf{d}_{t-1}(V)/7$. Thus, we have $\mathbf{d}_{t-1}(\bigcup_{A\in\mathcal{A}'_t}A)\ge 0.5\mathbf{d}_{t-1}(\bigcup_{A\in\mathcal{A}_t}A)$.
\end{proof}

With this, we have an implementation of~\Cref{ora: matching player} in ${O}(m\log^6(nW)\log \log nW)$ time, which we can already use to apply \Cref{thm: weak-decomp-del}. Next, we will also implement~\Cref{ora: grafting flow}, so we can apply \Cref{thm: weak-decomp-del2}.

\subsection{Grafting Deleted Nodes}\label{sec:prop-3}

In this section, we will use the pseudo-congestion-approximator to solve another approximate max flow problem needed to construct the weak expander decomposition,~\Cref{ora: grafting flow} from \Cref{sec:weak-expander-decomp}. We are using~\Cref{ora: grafting flow} to guarantee boundary-linkedness in our weak expander decomposition and obtain the full guarantees given in \Cref{thm: weak-decomp-del2}. 

Now, we restate the flow problem which we want to solve. Recall that in our setting, $\bd_T(u)$ is either $0$ or $\bd(u)$ for all $u\in V$, and $\psi=\Theta(1)$. 
\begin{itemize}
    \item Let $\mathcal{A}_T^+=\{C\in\mathcal{A}_T:\mathbf{d}_T(C)>0\}$ and $V^+=\bigcup_{C\in\mathcal{A}_T^+}C$.
    \item For each $C \in \mathcal{A}_T^+$:
    \begin{itemize}
        \item For each $u \in C$, add $\Delta(u) = \deg_{\partial C}(u)+ \bd(u) - \bd_T(u)$ source. 
        \item For each $u \in C$ with $\bd_T(u) = \bd(u)$, add sink $\nabla(u) = \bd(u)/5$. 
    \end{itemize}
    \item Remove all edges cut by $\mathcal{A}_T^+$ from $G$, and scale the capacity of remaining edges by $1/\psi$. 
\end{itemize}

Note that this flow problem on $G$ is the combination of independent flow problems in each component of $\mathcal{A}_T^+$. We wish to find a fair (one-sided) flow-cut pair, as defined in \Cref{sec:fair}, for this flow problem. In order to solve this, we first state our subroutine for one-sided fair cuts, which we prove in~\Cref{sec:fair}.

\begin{restatable}{theorem}{faircut}\label{thm:fair-cut}
Consider a graph $G=(V,E,\bc)$ with $\bc \in \Z \cap [1,W]$, a vertex subset $U\subseteq V$, and a vertex $t\in U$. Let $\mathcal C$ be a laminar family of vertex subsets of $V\setminus\{t\}$ of total size $z$, such that any demand vector $\mathbf b\in\mathbb R^V$ satisfying $|\mathbf b(C)|\le\delta C$ can be routed in $G$ with congestion $q$ in time $T$. There is an algorithm in time $O(z\log nW+\epsilon^{-1} qm\log^3nW+T\log nW)$ that computes a set $A\subseteq U$ containing $t$ and a flow $f$ such that
 \begin{enumerate}
 \item $\delta A\le4\delta U$.
 \item Each edge $(u,v)\in\partial A$ with $u\notin A,\,v\in A$ sends at least $(1-\epsilon)$ fraction capacity of flow into $A$.
 \item Each vertex $v\in A\setminus\{t\}$ carries a net flow of zero.
 \end{enumerate}
\end{restatable}

In order to simulate the independent flow problems on each component of $\mathcal{A}_T^+$, we need to define a slightly augmented graph $G^{\text{flow}}$ starting from $G=(V,E, \bc)$, on which we will define our flow problems. Start with the graph $G$. For each $C\in\mathcal{A}_T^+$ and each edge $e = (u,v)$ on the boundary $\partial C$, we add a new ``split node'' $x_e$, remove the edge $(u,v)$, and add two edges, $(u,x_e)$ and $(x_e,v)$, each with capacity $\bc(e)$. We also add a new node $t$, and, for each $C\in\mathcal{A}_T^+$, add an edge connecting every node $u\in C$ satisfying $\mathbf{d}_T(u)=\mathbf{d}(u)$ to $t$. 
This node $t$ will represent the sink in the flow problem, and we specify the rest of the capacities next.

We set the capacities of edges $(u,t)$ to be $\mathbf{d}(u)/5$, simulating the sink in the original flow problem. Then for each vertex $u\in V^+$ with $\bd_T(u)=0$, add a ``leaf node'' from $u$ denoted $\tilde{u}$ and an edge $(\tilde{u},u)$ with capacity $\mathbf{d}(u)-\mathbf{d}_T(u)=\bd(u)$. We ensure that $u\in U$ and $\tilde u\notin U$, which intuitively simulates a source of $\mathbf{d}(u)$ at each such $u\in V^+$. In particular, note that if, for example, $u \in A$, then  $(1 - \epsilon) \bc(e)$ flow is sent along $(\tilde u, u)$ from (2) of~\Cref{thm:fair-cut}. Furthermore, for each edge $e = (u,v) \in \partial C$ for $C \in \calA_T^+$, we ensure that $u,v\in U$ and $x_e\notin U$; since the capacity of $(u,x_e)$ and $(x_e,v)$ are both $\mathbf{c}(e)$ and these are the only edges incident to split node $x_e$, this intuitively simulates having source of $\mathbf{c}(e)$ on $u$ and $v$, with the additional property that no flow can be sent between components. 

For technical reasons, we also connect each leaf node $\tilde{u}$ to $t$ with an edge $(\tilde{u},t)$ of capacity $(\mathbf{d}(u)-\mathbf{d}_T(u))/5=\mathbf{d}(u)/5$. For similar technical reasons, we also add an edge connecting each vertex $u\not\in V^+$ to $t$ with an edge of capacity of $\mathbf{d}(u)/5$. All flow using these edges will be removed, so our solution is a valid flow for~\Cref{ora: grafting flow}, but we need these edges for routing some residual demands. 

Finally, to apply \Cref{thm:fair-cut}, we define $U=V\cup\{t\}$ (not including the additional leaf nodes or split nodes). Importantly, the total capacity of the boundary is still not too much larger than the total source.
\begin{claim} \label{cla:deltaUbd}
 We have $\delta_{G^{\text{flow}}} U\le 6\Delta(V)/5$.  
\end{claim}
\begin{proof}
Note that the boundary edges of $U$ consist of exactly the edges incident to leaf nodes and split nodes representing sources, and the edges $(\tilde{u},t)$ which were added for technical reasons.     
\end{proof}

Let $\mathcal{C}$ be defined as in \Cref{lem:congestion-approximator-guarantee}, so that $\mathcal{C}$ is a pseudo-congestion-approximator. We augment $\mathcal C$ to a congestion-approximator $\mathcal C^{\text{flow}}$ in $G^{\text{flow}}$ as follows. For each $C\in\mathcal C$, define $\widetilde C=C\cup\{\tilde u:u\in V^+\cap C\}\cup\{x_e:e=(u,v),u,v\in C\}$ as the set with the copy $\tilde u$ added (if it exists) for each $u\in C$, and split node $x_e$ added if both endpoints are in $C$. Now define
\[ \mathcal C^{\text{flow}}=\{\widetilde C:C\in\mathcal C\}\cup\{\{\tilde u\}:u\in V^+\}\cup\{\{x_e\}:\text{split node }x_e\}. \]
We show that $\mathcal{C}^{\text{flow}}$ has the required properties to apply \Cref{thm:fair-cut}. 

\begin{lemma}\label{lem:routing}
    $\mathcal{C}^{\text{flow}}$ is a laminar family of vertex subsets excluding $t$ of total size $z=O(nL)$ such that any demand vector $\mathbf{b}$ on $G^{\textup{flow}}$ satisfying $|\mathbf{b}(C)|\le\delta_{G^{\textup{flow}}}C$ for each $C\in\mathcal{C}^{\textup{flow}}$ can be routed with congestion $O(\alpha\beta L^2)$ in time $O(m\log^4(nW))$.
\end{lemma}
\begin{proof}
    The size of $\mathcal{C}^{\text{flow}}$ and the fact that it is a laminar family follows immediately by the definition of $\mathcal{C}$ and our construction of $\mathcal C^{\text{flow}}$. For the rest of the proof, we assume that there is no scaling (i.e., $\psi=1$) since scaling by $1/\psi=\Theta(1)$ only affects congestion by a constant factor.
    
    Consider a demand vector $\mathbf{b}\in\mathbb{R}^V$ satisfying $|\mathbf{b}(C^{\text{flow}})|\le\delta_{G^{\text{flow}}}C^{\text{flow}}$ for all $C^{\text{flow}}\in\mathcal C^{\text{flow}}$. First, route the at most $\delta_{G^{\text{flow}}}\{x_e\}=2\mathbf c(e)$ demand at each split node $x_e$ to an arbitrary endpoint, which changes each $\mathbf{b}(C)$ for $C\in\mathcal C$ by at most $2\delta_GC$. Next, route the at most $\delta_{G^{\text{flow}}}\{\tilde u\}\le1.2\bd(u)$ demand at each vertex $\tilde u$ to $u$, which changes each $\mathbf{b}(C)$ for $C\in\mathcal C$ by at most $1.2\bd(C)$. It follows that after routing demand out of all $x_e$ and $\tilde u$, replacing $\mathbf b$ by the new demand, we have $|\mathbf b(C)|\le\delta_{G^{\text{flow}}}C+2\delta_GC+1.2\bd(C)$, which we will now bound by $O(\delta_GC)$. First, by construction of $\mathcal C$, each set $C\in\mathcal C$ is a subset of a set $C'\in\overline{\mathcal P}_L$, so $\bd(C)=\deg_{\partial\overline{\mathcal P}_L}(C)=\deg_{\partial C'}(C)\le\delta_GC$. Also, by construction of $G^{\text{flow}}$, we have $\delta_{G^{\text{flow}}}C\le\delta_GC+ \bd(C) + \bd(C)/5\le2.2\delta_GC$. It follows that after routing demand out of all $x_e$ and $\tilde u$, we have $|\mathbf b(C)|\le\delta_{G^{\text{flow}}}C+2\delta_GC+1.2\bd(C)\le2.2\delta_GC+2\delta_GC+1.2\delta_GC=O(\delta_GC)$ for all $C\in\mathcal C$. Up to a constant factor scaling in congestion, we may assume that $|\mathbf b(C)|\le\delta_GC$ instead.

    Next, we route this demand using the pseudo-congestion-approximator $\mathcal C$, routing through split nodes $x_e$ instead of the original edge $e$ whenever necessary.
  
    By the routing properties of pseudo-congestion-approximators (the routing guaranteed by \Cref{lem:congestion-approximator-guarantee}), we can route demand $\mathbf{b}-\mathbf{b}'$ such that $\mathbf{b}'\le|\deg_{\partial\overline{\mathcal{P}}_L}|$ with congestion $O(\alpha\beta L^2)$. By \Cref{lem:routing-runtime}, this routing can be found in $O(m\log^4(nW))$ time. Next, recall that for every node $u\in V$, there is either an edge $(u,t)$ with capacity $\mathbf{d}(u)/5=\deg_{\partial\overline{\mathcal{P}}_{L}}(u)/5$ or two edges $(\tilde{u},t)$ and $(\tilde{u},u)$, both with capacity at least $\mathbf{d}(u)/5=\deg_{\partial\overline{\mathcal{P}}_{L}}(u)/5$.
    This implies that any demand $\mathbf{b}'\le|\deg_{\partial\overline{\mathcal{P}}_{L}}|$ can be trivially routed with congestion $5$ by going through $t$. Specifically, have each node $u$ send $\mathbf{b}'(u)$ flow through the edge $(u,t)$ 
    (where negative values mean the flow is sent from $t$ to $u$). Each edge used in this flow has congestion $5$, and this flow routes the demand $\mathbf{b}'$.

    Summing up all of the routings, the overall congestion is $O(\alpha\beta L^2)$.
\end{proof}

\begin{lemma}\label{lem:oracle2}
    Assuming the conditions of \Cref{thm:iterative-algorithm}, there is an algorithm implementing the guarantees of~\Cref{ora: grafting flow} in time $O(m\log^8(nW))$.
\end{lemma}
\begin{proof}
We apply \Cref{thm:fair-cut} on $G^{\text{flow}}$ with $U=V\cup\{t\}$ and $\epsilon=\floweps/2$, for the desired $\floweps$ from~\Cref{ora: grafting flow} to obtain a flow $f$. To obtain our final flow, we first remove all flow on edges $(\tilde{u},t)$ or $(v,t)$ for $v\not\in V^+$. Next, take a flow decomposition of the resulting flow and only keep flow paths starting on the boundary $\partial A$ entering $A$. Finally, we truncate each flow path to end at $t$ to obtain our final flow $f'$. This has the desired runtime by~\Cref{thm:fair-cut} and~\Cref{lem:routing} (recall that $L = O(\log nW), \alpha = O(\log^3 nW),$ and $\beta = O(1)$). 

Next, we prove the two properties which~\Cref{ora: grafting flow} needs to satisfy. For property (1), consider each component $C\in\mathcal{A}_T$ satisfying $\deg_{\partial\mathcal{A}_T^+}(C)\le \mathbf{d}_T(C)/8$. We define the pair $(C_A,C\setminus C_A)$ to be $C_A=C\setminus A$, where $A$ is the set output by \Cref{thm:fair-cut}. For each $u\in C\setminus C_A$, we have that $u\in A$. 
    
    For property (1a), we wish to argue that we route at least $(1-\floweps)\Delta(u)$ source from $u$ to  $A\setminus C_A$. For each edge $e= (u,v)\in\partial C$, we add a split node $x_e$ and an edge $(u,x_e)$ with capacity $\mathbf{c}(e)$. If $\bd_T(u)=0$, we also add a leaf node $\tilde{u}$ along with an edge $(\tilde{u},u)$ of capacity $\mathbf{d}(u)-\mathbf{d}_T(u)$. The total capacity of these two types of edges is exactly $\Delta(u)$, by definition. Furthermore, both of these types of edges are in the boundary $\partial A$ because split nodes $x_e$ and leaf nodes $\tilde{u}$ are not in $U$ by definition, so they are also not in $A$. Hence, by property (2) of \Cref{thm:fair-cut}, the flow $f$ (and thus $f'$) sends at least $(1-\floweps)\Delta(u)$ flow from each $u\in C\setminus C_A$, satisfying property (1a) of the oracle.

    For property (1b), we wish to argue that the flow $f'$ saturates a $(1-\floweps)$ fraction of the capacity of each edge from $C_A$ to $C\setminus C_A$. Again, note that $C\setminus C_A\subseteq A$ and $C_A\cap A=\emptyset$, so each edge from $C_A$ to $C\setminus C_A$ lies in the boundary $\partial A$. Property (2) of \Cref{thm:fair-cut} again guarantees that the flow saturates a $(1-\floweps)$ fraction of the capacity of each such edge.
    
    For property (1c), we wish to argue that the cut $(C_A,C\setminus C_A)$ is small (on average). By~\Cref{cla:deltaUbd}, $\delta U=6\Delta(V)/5\le 2\mathbf{d}(V)$. Furthermore, we know that the capacity of each edge in the flow problem we are simulating is scaled up by $1/\psi$, and the edges in $\delta A$ are a superset of those in the boundary of cuts $(C_A,C\setminus C_A)$ because $\delta A$ may also contain edges to leaf nodes. All together, this implies that $$\sum_{C}\mathbf{c}_G(C_A,C\setminus C_A)\le \psi\cdot\delta A\le 4\psi\cdot\delta U\le 8\psi\cdot\mathbf{d}(V),$$ as desired. The first and third inequality follow from the above discussion, and the second inequality follows from property (1) of \Cref{thm:fair-cut}.

    Finally, for property (2), we wish to bound $\sum_{C \in \calA_T^+}\mathbf{d}(C_A)$. Note that, by definition of $\nabla(C_A)$, we have 
    \[
    \nabla(C_A) = \bd_T(C_A)/5 \geq \frac{15}{16} \cdot \frac{\bd(C_A)}{5} = \frac{3 \bd(C_A)}{16}.
    \]
    The first inequality uses that for each $C \in \calA_T^+$, we have $\bd_T(C) > 15 \bd(C)/16$.
     But observe that each edge from $C_A$ to $t$ is cut, because $t\not\in C_A=C\setminus A$, which implies that $$\sum_{C}\nabla(C_A)\le\delta A\le 4\delta U \leq 5\cdot (\mathbf{d}(V)-\mathbf{d}_T(V)+\deg_{\partial\mathcal{A}_T^+}(V)).$$
     The last inequality uses~\Cref{cla:deltaUbd} and the definition of $\Delta(V)$. Combining gives that $\mathbf{d}(\bigcup_{C}C_A)\le 30\cdot(\mathbf{d}(V)-\mathbf{d}_T(V)+\deg_{\partial\mathcal{A}_T}(V))$
\end{proof}

\subsection{Proof of \Cref{thm:iterative-algorithm}}\label{sec:combining-everything}

In this previous two subsections, we have shown how to implement~\Cref{ora: matching player,ora: grafting flow}, so we can apply \Cref{thm: weak-decomp-del2}. We restate it here for convenience:

\maindecomp*

Now, we prove \Cref{thm:iterative-algorithm}. Choose $\floweps_1=T/10^{10}$, $\floweps_2=1/10^{10}$, $\phi=\log(nW)/10^{10}$, $\psi=1/10^{10}$, and define $\mathbf{d}=\deg_{\partial\overline{\mathcal{P}}_L}$. Assuming the conditions given in \Cref{thm:iterative-algorithm}, we have given an algorithm implementing~\Cref{ora: matching player} in time $R_1(m,n,\floweps_1)=O(m\log^6(nW)\log\log(nW))$ and~\Cref{ora: grafting flow} in time $R_2(m,n,\floweps_2)=O(m\log^8(nW))$ (see \Cref{lem:oracle1,lem:oracle2}). Hence, we can apply \Cref{thm: weak-decomp-del2}, to obtain a partition $\mathcal{A}=\mathcal{A}^\circ\sqcup\mathcal{A}^\times$ of $V$ with properties (1)--(5) in time $O(m\log^8(nW)\log\log(nW))$.

We define $\mathcal{P}_{L+1}=\mathcal{A}^\circ$, $V_{L+1}=\bigcup_{A\in\mathcal{A}^\circ}A$, and extend $\mathcal{P}_{L+1}$ from a partition of $V_{L+1}$ to a partition of $V$ by defining $\mathcal{Q}_{L+1}=\{C\cap (V\setminus V_{L+1}):C\in\overline{\mathcal{P}}_L\}$ and $\overline{\mathcal{P}}_{L+1}=\mathcal{P}_{L+1}\sqcup\mathcal{Q}_{L+1}$. We now verify properties (2), (3), and (4) in \Cref{thm:iterative-algorithm}. Property (2) follows by property (4) from \Cref{thm: weak-decomp-del2}. Property (3) follows from property (5) from \Cref{thm: weak-decomp-del2}. Finally, consider property (4). The capacity of edges cut by $\mathcal{P}_{L+1}$ is upper bounded by the capacity of edges cut by $\mathcal{A}$ and the capacity of edges cut by $\mathcal{Q}_{L+1}$ and not already cut by $\mathcal{P}_{L+1}$ is upper bounded by $\deg_{\partial\overline{\mathcal{P}}_L}(\bigcup_{A\in\mathcal{A}^\times}A)=\mathbf{d}(\bigcup_{A\in\mathcal{A}^\times}A)$. By properties (2) and (3) from \Cref{thm: weak-decomp-del2}, the claimed bound follows.

\subsubsection{Justifying Oracle Runtimes}

In the above proof, we assumed that the routings guaranteed by properties (2) and (3) of \Cref{thm:iterative-algorithm} can be computed in $\mathcal{T}_{2}=O(m\log^3(nW))$ and $\mathcal{T}_3=O(m\log(nW))$ time, respectively. We explain how to obtain these routings here. For property (2), fix $i\in[L-1]$ and consider some demand $\mathbf{d}\in\mathbb{R}^V$ satisfying $|\mathbf{d}|\le \sum_{C\in\overline{\mathcal{P}}_{i+1},C\subseteq V_{i+1}}\deg_{\partial\overline{\mathcal{P}}_{i}\cup\partial C}|_C$. We already show in \Cref{rmk:-fast-mixing} that any demand $|\mathbf{b}|\le\mathbf{d}_T$ can be routed in $O(m\log^3(nW))$ time. Thus, it suffices to route $\mathbf{d}$ such that the residual $\mathbf{b}$ satisfies $|\mathbf{b}|\le\mathbf{d}_T$. To do this, we use the flow $f$ obtained from \Cref{ora: grafting flow} to route $\mathbf{d}-\mathbf{b}$ as described in the proof of \Cref{thm: weak-decomp-del2}. This takes $O(m\log(nW))$ time using link-cut trees to implement the flow-path decomposition manipulations~\cite{DBLP:journals/jcss/SleatorT83}. For property (3), we obtain the flow explicitly using \Cref{thm: weak-decomp-del2} property (5) in each previous level $i\in[L-1]$, so we can obtain this in $\mathcal{T}_3=O(m\log(nW))$ time.

\section*{Acknowledgments}
This material is based upon work supported by the National Science Foundation Graduate
Research Fellowship Program under Grant No.\ DGE2140739. Any opinions, findings,
and conclusions or recommendations expressed in this material are those of the author(s)
and do not necessarily reflect the views of the National Science Foundation.

\bibliographystyle{alpha}
\bibliography{refs}

\appendix

\section{Faster Algorithm for One-Sided Fair Cuts}\label{sec:fair}

In this section, we prove the following fair cuts routine, following the general strategy of \cite{DBLP:conf/ipco/LiL25}. The main difference is that we have a specific congestion approximator to work with, which introduces some technical difficulties. We actually obtain what is called a \emph{one-sided fair cut} from \cite{DBLP:conf/soda/0006NPS23}, a weaker object but sufficient for our purposes.
\faircut*

For the rest of this section, we prove \Cref{thm:fair-cut}. We use the following instantiation of Sherman's algorithm~\cite{sherman2017area} as stated in~\cite{DBLP:conf/ipco/LiL25}, but with the running time speedup from~\cite{jambulapati2023revisiting} (see Section~5 of their arXiv version). We also implicitly use the fact that the necessary matrix-vector multiplications from~\cite{jambulapati2023revisiting} can be done in $O(m)$ time since $\mathcal C$ is assumed to be laminar.

\begin{theorem}[see Theorem~4 of~\cite{DBLP:conf/ipco/LiL25}]\label{thm:sherman}
Consider a graph $G=(V,E,\bc)$, a residual graph $G'$ of $G$, two vertices $s,t\in V$, parameters $\epsilon,\tau>0$, and let $\mathcal C$ be a laminar family of vertex subsets. There is an $O(\epsilon^{-1}m\log nW)$ time algorithm that computes either
 \begin{enumerate}
 \item An $(s,t)$-cut in $G'$ of value less than $\tau$, or
 \item A flow $f$ in $G'$ routing a demand $\mathbf d$ such that the residual demand $\tilde{\mathbf d}=\tau(\mathbf 1_s-\mathbf 1_t)-\mathbf d$ satisfies $|\tilde{\mathbf d}(C)|\le\epsilon\delta C$ for all $C\in\mathcal C$.
 \end{enumerate}
\end{theorem}

Throughout the algorithm, we use a more convenient parameter $\epsilon'=\Theta(\epsilon/\log n)$.
The algorithm initializes $A^{(0)}\gets U$ and $f^{(0)}\gets\mathbf 0\in\mathbb R^E$ as the empty flow, and proceeds for a number of iterations. On a given iteration $k$, let $\vec G^{(k)}$ be the residual graph of flow $f^{(k)}$.  The algorithm first computes a set $B^{(k)}\subseteq A^{(k)}$ as follows. Initialize $B^{(k)}\gets A^{(k)}$ and iterate through the sets $C\in\mathcal C$ in decreasing order of set size $|C|$. For each set $C\in\mathcal C$ in this order, if
\[ \sum_{v\in B^{(k)}\cap C}\sum_{u:(u,v)\in\vec\partial B^{(k)}}\bc_{\vec G^{(k)}}(u,v)>2\delta_GC ,\]
then the algorithm updates $B^{(k)}\gets B^{(k)}\setminus C$. To implement this checking efficiently, we maintain the value $\sum_{u:(u,v)\in\vec\partial B^{(k)}}\bc_{\vec G^{(k)}}(u,v)$ at each vertex $v\in V$. Checking the sets $C\in\mathcal C$ takes $O(z)$ time total, and each update to the value at $v$ is the result of removing a neighbor $u$ from $B^{(k)}$, so there are at most $\deg(v)$ many updates to $v$ and at most $2m$ updates total.
\begin{lemma}\label{lem:bounding-C}
At the end of this procedure, for each set $C\in\mathcal C$,
\[ \sum_{v\in B^{(k)}\cap C}\sum_{u:(u,v)\in\vec\partial B^{(k)}}\bc_{\vec G^{(k)}}(u,v)\le4\delta_GC .\]
\end{lemma}
\begin{proof}
Immediately after $C$ is processed, we must have
\[ \sum_{v\in B^{(k)}\cap C}\sum_{u:(u,v)\in\vec\partial B^{(k)}}\bc_{\vec G^{(k)}}(u,v)\le2\delta_GC .\]
(The left expression is $0$ if the algorithm updated $B^{(k)}\gets B^{(k)}\setminus C$.) Consider the potential function
\[ \Phi(C)=\sum_{(u,v)\in\vec\partial(B^{(k)}\cap C)}\bc_{\vec G^{(k)}}(u,v) \]
immediately after processing $C$, which is bounded by
\begin{align*}
\sum_{(u,v)\in\vec\partial(B^{(k)}\cap C)}\bc_{\vec G^{(k)}}(u,v)&\le\sum_{v\in B^{(k)}\cap C}\sum_{u:(u,v)\in\vec\partial B^{(k)}}\bc_{\vec G^{(k)}}(u,v)+\sum_{(u,v)\in\vec\partial C}\bc_{\vec G^{(k)}}(u,v) \label{eq:bounding-C}
\\&\le\sum_{v\in B^{(k)}\cap C}\sum_{u:(u,v)\in\vec\partial B^{(k)}}\bc_{\vec G^{(k)}}(u,v)+2\delta_GC
\\&\le4\delta_GC.
\end{align*}
We now claim that $\Phi(C)$ can only decrease upon processing later sets $C'\in\mathcal C$. Observe that it can only change if $C'\cap C\ne\emptyset$. Since $\mathcal C$ is a laminar family that we iterate in decreasing order of size, any such set $C'$ is contained in $C$. We can view updating $B^{(k)}\gets B^{(k)}\setminus C'$ for some $C'\subseteq C$ as removing the contribution of $\bc_{\vec G^{(k)}}(u,v)$ from edges $(u,v)\in\vec\partial B^{(k)}$ with $v\in B^{(k)}\cap C'$, which are also edges in the summation in $\Phi(C)$, and then adding back a subset of $\delta_GC$ to form the new boundary, which introduces a total contribution of at most $2\delta_GC$ to the summation (where the factor $2$ is because the residual capacity can be at most twice the original). It follows that $\Phi(C)$ can only decrease upon processing $C'$. Since $\Phi(C)\le4\delta_GC$ upper bounds the expression from the lemma statement, the proof follows.
\end{proof}

Next, we build the following graph on which to run \Cref{thm:sherman}. Let $H^{(k)}$ be the undirected graph consisting of $G[B^{(k)}]$ and a new vertex $s$ with the following undirected edges: for each vertex $v\in B^{(k)}$, add the edge $(s,v)$ of capacity
\[ \bc_{H^{(k)}}(s,v)=\frac12\sum_{u:(u,v)\in\vec\partial B^{(k)}}\bc_{\vec G^{(k)}}(u,v) .\]
Let $\vec H^{(k)}$ be the directed graph consisting of $\vec G^{(k)}[B^{(k)}]$ and a new vertex $s$, with directed edges $(s,v)$ of capacity $2\bc_{H^{(k)}}(s,v)$. It is clear that $\vec H^{(k)}$ is a residual graph of $H^{(k)}$: take the flow $f^{(k)}$ restricted to the edges of $G[B^{(k)}]$ and add $\bc_{H^{(k)}}(s,v)$ flow along each edge $(s,v)$ in the direction from $v$ to $s$. Apply \Cref{thm:sherman} on graph $H^{(k)}$, residual graph $\vec H^{(k)}$, vertices $s,t$, parameters $\epsilon=\epsilon'/q$ and $\tau=\deg_{H^{(k)}}(s)$, and the family of sets
\[ \mathcal C^{(k)}=\{C\cap B^{(k)}:C\in\mathcal C\}\cup\{\{s\}\} .\]
There are now two cases:
 \begin{enumerate}
 \item If \Cref{thm:sherman} outputs a cut, let $S$ be the side containing $s$. Update $A^{(k+1)}=B^{(k)}\setminus S$ and $f^{(k+1)}=f^{(k)}$, i.e., update the cut and leave the flow unchanged.
 \item If \Cref{thm:sherman} outputs a flow $f$, then we first map the flow to a feasible flow $f'$ in $\vec G^{(k)}$ by splitting the flow across edges $(s,v)$ into flow across edges $(u,v)\in\vec\partial B^{(k)}$. This can always be done because each edge $(s,v)$ has capacity $2\bc_{H^{(k)}}(s,v)=\sum_{u:(u,v)\in\vec\partial B^{(k)}}\bc_{\vec G^{(k)}}(u,v)$.
 Then, update $f^{(k+1)}=f^{(k)}+f'$ and $A^{(k+1)}=B^{(k)}$, i.e., update the flow and leave the cut unchanged from $B^{(k)}$.
 \end{enumerate}

We first prove a lemma about the family $\mathcal C^{(k)}$.
\begin{lemma}\label{lem:boundary-of-H}
$\delta_{H^{(k)}}(C\cap B^{(k)})\le3\delta_GC$ for all $C\in\mathcal C$.
\end{lemma}
\begin{proof}
For each $C\in\mathcal C$, the only edges in $\partial_{H^{(k)}}(C\cap B^{(k)})\setminus\partial_GC$ are those incident to $s$. By the construction of $H^{(k)}$, the total capacity of edges between $s$ and $C$ is
\[ \frac12\sum_{v\in B^{(k)}\cap C}\sum_{(u,v)\in\vec\partial B^{(k)}}\bc_{\vec G^{(k)}}(u,v) ,\]
which is at most $2\delta_GC$ by \Cref{lem:bounding-C}. It follows that $\delta_{H^{(k)}}(C\cap B^{(k)})\le3\delta_GC$.
\end{proof}

To measure progress over the iterations, consider the potential function
\[ \Phi^{(k)}=\sum_{(u,v)\in\vec\partial A^{(k)}}\bc_{\vec G^{(k)}}(u,v). \]
\begin{lemma}\label{lem:potential-decrease}
$\Phi^{(k+1)}\le\frac34\Phi^{(k)}$ for all iterations $k\ge1$.
\end{lemma}
\begin{proof}
We first show that on iteration $k$, the parameter $\tau$ in the call to \Cref{thm:sherman} satisfies $\tau\le\frac12\Phi^{(k)}$. We first claim that
\begin{gather}
\sum_{(u,v)\in\vec\partial B^{(k)}}\bc_{\vec G^{(k)}}(u,v)\le\sum_{(u,v)\in\vec\partial A^{(k)}}\bc_{\vec G^{(k)}}(u,v)=\Phi^{(k)} .\label{eq:bound-boundary-B}
\end{gather}
This is because we can view each update $B^{(k)}\gets B^{(k)}\setminus C$ as removing more than $2\delta_GC$ from the summation on the left of~(\ref{eq:bound-boundary-B}), and then adding at most $2\delta_GC$ to it (where the factor $2$ is because the residual capacity can be at most twice the original). We conclude that
\[ \tau=\deg_{H^{(k)}}(s)=\frac12\sum_{(u,v)\in\vec\partial B^{(k)}}\bc_{\vec G^{(k)}}(u,v)\le\frac12\Phi^{(k)} .\]
We now consider the two possible outputs of \Cref{thm:sherman}.
 \begin{enumerate}
 \item If \Cref{thm:sherman} outputs a cut, then the new value of $\Phi^{(k)}$ is exactly the capacity of this cut in $\vec H^{(k)}$, which is at most $\tau\le\frac12\Phi^{(k)}$. Here, we use the fact that the edges incident to $s$ in the cut map to edges in $\vec\partial B^{(k+1)}$ with the same total capacity.
 \item If \Cref{thm:sherman} outputs a flow $f$, then since $\{s\}\in\mathcal C^{(k)}$, the flow out of vertex $s$ is at least
\[ \tau-\frac{\epsilon'}q\deg_{H^{(k)}}(s)\ge\tau-\frac12\deg_{H^{(k)}}(s)=\frac12\tau=\frac14\sum_{(u,v)\in\vec\partial B^{(k)}}\bc_{\vec G^{(k)}}(u,v) .\]
Therefore, the mapped flow $f'$ sends at least that much flow across $\vec\partial B^{(k)}$. By the property of residual graphs, the new residual capacity across $\vec\partial B^{(k)}=\vec\partial A^{(k+1)}$ is at most $3/4$ fraction of the old residual capacity, which is at most $\frac34\Phi^{(k)}$ by (\ref{eq:bound-boundary-B}).\qedhere
 \end{enumerate}
\end{proof}

Therefore, after $O(\log nW)$ iterations, we have $\Phi^{(k)}\le1/\textup{poly}(n)$, which means there is (effectively) no remaining residual capacity along $\vec\partial A^{(k)}$. By standard flow rounding, we can even ensure no remaining residual capacity, in which case the final flow $f^{(k)}$ sends full capacity along each edge in $\vec\partial A^{(k)}$. However, this flow may not satisfy property~(3) of \Cref{thm:fair-cut}, so we need to route the residual demands.
\begin{lemma}
For each iteration $k$, if \Cref{thm:sherman} outputs a flow $f$ which is mapped to $f'$ in $\vec G^{(k)}$, then we can compute a flow $f''$ in $G$ of congestion $3\epsilon'$ such that in the flow $f'+f''$, each vertex $v\in B^{(k)}\setminus\{t\}$ carries a net flow of zero, and vertex $t$ receives at most $\Phi^{(k)}$ flow.
\end{lemma}
\begin{proof}
Suppose that $f'$ routes demand $\mathbf d'$ in $\vec G^{(k)}$. Since net flow from vertices in $B^{(k)}$ do not change when mapping from $f$ to $f'$, the demand $\mathbf d'$ still satisfies $|\mathbf d'(C\cap B^{(k)})|\le\frac{\epsilon'}q\delta_{H^{(k)}}(C\cap B^{(k)})\le\frac{3\epsilon'}q\delta_GC$ for all $C\in\mathcal C$, where the last inequality holds by \Cref{lem:boundary-of-H}. Let $\mathbf d''$ be demand $\mathbf d'$ with all coordinates outside $B^{(k)}$ zeroed out and with demand at $t$ modified to ensure that $\mathbf d''(V)=0$. For each $C\in\mathcal C$, since $C$ does not contain $t$, we have $|\mathbf d''(C)|=|\mathbf d'(C\cap B^{(k)})|\le\frac{3\epsilon'}q\delta_GC$. By the assumption of $\mathcal C$ scaled by $q$, we can route $\mathbf d''$ on $G$ with congestion $3\epsilon'$ in time $T$, obtaining our desired flow $f''$.

It remains to show that vertex $t$ receives at most $\Phi^{(k)}$ flow. If we added flow $f''$ onto $f$ instead, where the new flow is allowed to use edges of $\vec G^{(k)}$ outside $B^{(k)}$, then it is an $(s,t)$-flow. This flow is feasible, so it has value at most the sum of capacities of edges out of $s$, which is at most $\Phi^{(k)}$ by (\ref{eq:bound-boundary-B}).
\end{proof}
The algorithm uses the previous lemma on each of the $O(\log nW)$ flows $f'$ seen throughout the iterations. The sum of the flows $f''$ has congestion $O(\epsilon'\log nW)\le\epsilon$, and adding this flow onto $f^{(k)}$ fulfills property~(3) of \Cref{thm:fair-cut} while ensuring each edge still sends at least $(1-\epsilon)$ fraction capacity into $A$, fulfilling property~(2).

Finally, by the previous lemma, vertex $t$ receives a total of at most $\sum_k\Phi^{(k)}$ flow over the iterations. By \Cref{lem:potential-decrease}, this summation is at most $4\Phi^{(0)}=4\delta U$, fulfilling property~(1).

For the final running time, each iteration is dominated by the $O(z)$ time to compute $B^{(k)}$, the single call to \Cref{thm:sherman} with error parameter $\epsilon'/q=\Theta(\frac\epsilon{q\log nW})$, and using $\mathcal C$ to route a demand in time $T$. There are $O(\log nW)$ iterations, so the total running time is
\[ O(\log nW)\cdot O(z+\epsilon^{-1}q\log nW\cdot m\log nW+T)=O(z\log nW+\epsilon^{-1} qm\log^3nW+T\log nW) ,\]
as promised by \Cref{thm:fair-cut}.

\section{Omitted Proofs} \label{sec: omit}
In this section, we prove the results whose proofs were omitted from the main body.

\subsection{Proof of~\Cref{lem: progress set}}
We restate~\Cref{lem: progress set} for convenience.
\progset*
\begin{proof}
If $|X| = 2$, the result is trivial; set $\eta = \bar \mu$. 

Otherwise, for any $c$, define $A_{c} := \{x \in X \,:\, x < c\}$ and $B_{c} := \{x \in X \,:\, x \geq c\}$. Observe that we have 
\[
\sum_{x \in A_{\bar{\mu}}} \bar{\mu} - x = \sum_{x \in B_{\bar \mu}} x - \bar \mu =: d.
\]
Let $\eta$ be such that $|B_\eta| = r|X|$, where $r = \frac{\left \lceil |X|/8| \right \rceil}{|X|}$. Suppose now that $\bar \mu \geq \eta$. The other case will be analogous. 
First, if 
\begin{equation} 
\sum_{x \in B_{\bar \mu }}(x - \bar \mu)^2 \geq \frac{1}{36} \sum_{x \in X} (x - \bar \mu)^2, \label{eqtn: case 1}
\end{equation}
then, in particular, 
\[
\sum_{x \in B_{\eta}}(x - \bar \mu)^2 \geq \frac{1}{36} \sum_{x \in X} (x - \bar \mu)^2,
\]
so we can output $B_{\eta}$ as $L_{\eta}$, with $S = L_{\eta}$. 
As such, assume that~\Cref{eqtn: case 1} does not hold. Hence, we have 
\begin{equation}
\sum_{x \in A_{\bar \mu}} (x-\bar \mu)^2 > \frac{35}{36} \sum_{x \in X} (x - \bar \mu)^2. \label{eqtn: case 2}
\end{equation}
Since $\bar \mu \geq \eta$ the elements in $A_{\bar \mu} \setminus A_{\eta}$ satisfy 
\begin{align*}
\sum_{x \in A_{\bar \mu} \setminus A_{\eta}} (x - \bar \mu)^2 &\leq |A_{\bar \mu} \setminus A_{\eta}| (\eta - \bar \mu)^2 \\
&\leq r|X| (\eta - \bar \mu)^2 \tag{Since $A_{\bar \mu} \setminus A_{\eta} \subseteq B_\eta$} \\
&\leq \frac{r}{(1-r)} \sum_{x \in A_\eta} (x - \bar \mu)^2. \tag{Since $|A_\eta| \geq (1-r)|X|$ and  $\bar \mu \geq \eta$}
\end{align*}
Combining the above with~\Cref{eqtn: case 2}, we can conclude that
\begin{equation} \label{eqtn: lb Aeta}
    \sum_{x \in A_{\eta}} (x - \bar \mu)^2 > \frac{35}{36} \left(1-\frac{r}{1-r}\right)\sum_{x \in X} (x - \bar \mu)^2. 
\end{equation}
Now, let $\eta' \leq \eta$ such that $|A_{\eta'}| = r|X|$. By~\Cref{eqtn: lb Aeta} and the fact that $|A_{\eta'}| = r|X|$ and $|A_{\eta}| = (1-r)|X|$, we get 
\begin{equation}
\sum_{x \in A_{\eta'}} (x - \bar \mu)^2 >  \frac{35r}{36(1-r)}\left(1-\frac{r}{1-r}\right)\sum_{x \in X} (x - \bar \mu)^2.    \label{eqtn: lb Aeta'}
\end{equation}
In this case, we will output $A_{\eta'}$ as $L_{\eta}$. It remains to find $S \subseteq A_{\eta'}$ with the desired properties. Crucially, observe that $\eta' \geq \bar \mu - \frac{d}{r|X|}$. Indeed, otherwise we have 
\begin{align*}
    d &= \sum_{x \in A_{\bar \mu}} \bar \mu - x \\
    &\geq \sum_{x \in A_{\eta'}} \bar \mu - x \tag{Since $A_{\eta'} \subseteq A_{\bar \mu}$} \\
    &> |A_{\eta'}| \frac{d}{r|X|} \tag{By assumption, $\eta' < \bar \mu - \frac{d}{r|X|}$} \\
    &= d. \tag{As $|A_{\eta'}| = r|X|$}
\end{align*}
Consequently, for any $x \in X$ with $\bar \mu - x = \frac{C_rd}{|X|}$ for $C_r \geq \frac{3}{2r}$, the second property of $S$ is satisfied:
\begin{align*}
    (\eta' - x)^2 &\geq (\bar \mu - \frac{d}{r|X|} - x)^2 \\
    &= \frac{(C^2_r - 1/r)^2d^2}{|X|^2} \\
    &\geq \frac{C^2_rd^2}{9|X|^2} = (\bar \mu - x)^2/9.
\end{align*}
Let $I := \{x \in A_{\eta'} \,:\, 0 < \bar \mu - x < \frac{3d}{2r|X|}\}$ be the set of intermediate elements  in $A_{\eta'}$ that do not satisfy the second property of $S$ by the above. Note that 
\[
\sum_{x \in I} (\bar \mu - x)^2 \leq \frac{9d^2 |I|}{4r^2|X|^2} \leq \frac{ 9d^2}{4r|X|},
\]
using that $I \subseteq  A_{\eta'}$ and $|A_{\eta'}| = r|X|$.
On the other hand, 
\begin{align*}
\sum_{x \in B_{\bar \mu}} (\bar \mu - x)^2 &\geq  \frac{\left(\sum_{x \in B_{\bar \mu}} (\bar \mu - x)\right)^2}{|B_{\bar \mu}|}\tag{By the Cauchy-Schwarz inequality} \\
&= \frac{d^2}{|B_{\bar \mu}|} \geq \frac{d^2}{r|X|}. \tag{Since $|B_{\bar \mu}| \leq r|X|$ because $\bar \mu \geq \eta$}
\end{align*}
Since we are assuming~\Cref{eqtn: case 1} does not hold, we get
\begin{align*}
\frac{1}{36} \sum_{x \in X} (x - \bar \mu)^2 &> \sum_{x \in B_{\bar \mu}} (\bar \mu - x)^2 \tag{\Cref{eqtn: case 1} does not hold}\\
&\geq \frac{d^2}{r|X|} \\
&\geq \frac{4}{9} \sum_{x \in I} (\bar \mu - x)^2.
\end{align*}
So, we get
\[
\sum_{x \in I} (\bar \mu - x)^2 \leq \frac{9}{144} \sum_{x \in X} (x - \bar \mu)^2.
\]
Hence, we have, by~\Cref{eqtn: lb Aeta'},
\begin{align*}
 \sum_{x \in  A_{\eta'} \setminus I} ( x - \bar \mu)^2 &\geq \left(\frac{35r}{36(1-r)} \left(1-\frac{r}{1-r}\right)\left(1-\frac{1}{z}\right) - \frac{9}{144}\right)\sum_{x \in X} (x - \bar \mu)^2 \\
 &\geq \left(\frac{35}{360} - \frac{9}{144}\right)\sum_{x \in X} (x - \bar \mu)^2 \tag{Since $|X| \geq 3$, $1/8 \leq r \leq 3/8$} \\
 &\geq \frac{1}{36} \sum_{x \in X} (x - \bar \mu)^2.
\end{align*}
Then $S := A_{\eta'} \setminus I$ satisfies the desired properties.
\end{proof}

\end{document}